\documentclass[11pt]{article}

\usepackage[utf8]{inputenc}
\usepackage{amsmath,amsthm,amsfonts,amssymb,dsfont}
\usepackage{natbib}
\usepackage[top=2cm,left=3cm,right=3cm,bottom=2cm]{geometry}
\usepackage{textcomp}
\usepackage{enumerate}
\usepackage{comment}
\usepackage{color}
\usepackage{booktabs} 
\usepackage{rotating}
\usepackage{caption}
\usepackage{pdflscape}
\usepackage{subcaption}		
\usepackage{soul}

\newcounter{hypA}

\newcommand{\R}{\mathbb{R}}

\usepackage{enumitem}
\usepackage{enumerate}
\usepackage{float}
\usepackage{multirow}
\usepackage{graphicx}
\usepackage{adjustbox}
\usepackage{diagbox}
\usepackage{rotating}
\usepackage{adjustbox}
\usepackage{tabularray}
\usepackage{tikz}
\usepackage{tikz-3dplot}

\usepackage{tikz}
\usetikzlibrary{shapes.geometric, arrows}
\tikzstyle{startstop} = [rectangle, rounded corners, minimum width=3cm, minimum height=1.5cm,text centered, draw=black, fill=white!70]
\tikzstyle{io} = [trapezium, trapezium left angle=70, trapezium right angle=110, minimum width=3cm, minimum height=1cm, text centered, draw=black, fill=black!50]
\tikzstyle{process} = [rectangle, minimum width=3cm, minimum height=1cm, text centered, draw=black, fill=red!40]
\tikzstyle{decision} = [diamond, minimum width=3cm, minimum height=1cm, text centered, draw=black, fill=black!10]
\tikzstyle{arrow} = [thick,->,>=stealth]

\tikzstyle{startstop2} = [rectangle, rounded corners, minimum width=3cm, minimum height=1.5cm,text centered, draw=black, fill=orange!70]
\tikzstyle{io2} = [trapezium, trapezium left angle=70, trapezium right angle=110, minimum width=3cm, minimum height=1cm, text centered, draw=black, fill=orange!50]
\tikzstyle{process2} = [rectangle, minimum width=3cm, minimum height=1cm, text centered, draw=black, fill=orange!40]
\tikzstyle{decision2} = [diamond, minimum width=3cm, minimum height=1cm, text centered, draw=black, fill=orange!10]
\tikzstyle{arrow} = [thick,->,>=stealth]
\usepackage{cancel,xcolor}
\newcommand{\Cancel}[2][black]{{\color{#1}\cancel{\color{black}#2}}}

\newtheorem{theorem}{Theorem}
\newtheorem{proposition}{Proposition}
\newtheorem{corollary}{Corollary}
\newtheorem{lemma}{Lemma}
\newtheorem{remark}{Remark}
\newtheorem{assumptions}{A}

\begin{document}

\title{
Discriminant Analysis in stationary time series based on robust cepstral coefficients}

\author{Jonathan de Souza Matias and Valderio Anselmo Reisen}

\date{\today}

\maketitle

    \subsection*{Abstract}

        Time series analysis is essential in fields such as finance, economics, environmental science, and biomedical engineering for understanding underlying mechanisms, forecasting, and identifying patterns. Traditional time domain methods, which focus on trends, seasonality, and noise, often overlook periodicities and harmonic structures that are better captured in the frequency domain. Analyzing time series in the frequency domain enables the identification of these spectral properties, providing deeper insights into the underlying processes. These insights can help differentiate data-generating processes of different populations and assist in the discrimination and classification of time series. The literature commonly uses smoothed estimators like the smoothed periodogram to minimize bias, obtaining an average spectrum from individual replicates within a population to classify new time series. However, if there is spectral variability among replicates within each population, such methods become unfeasible. Moreover, abrupt values can significantly impact spectrum estimators, complicating practical discrimination and classification. There is a gap in the literature for methods that consider within-population spectral variability, separate white noise effects from autocorrelations, and use robust estimators in the presence of outliers. This paper addresses this gap by presenting a robust framework for classifying replicate groups of time series by transforming them into the frequency domain using the Fourier Transform to compute the power spectrum. Then, after taking the logarithm of the spectra, the inverse Fourier Transform is used to achieve the cepstrum. To mitigate the effects of outliers and improve the consistency of spectral estimates, we employ the multitaper periodogram alongside the M-periodogram. These spectral features are then utilized in Linear Discriminant Analysis (LDA) to enhance classification accuracy and interpretability. This integrated approach offers significant potential for applications requiring precise temporal pattern distinction and resilience to data anomalies.
    
        \subsection*{Keywords}
        Robust, discriminant, classification, cepstral, multitaper, time series, frequency domain.
    \section{Introduction}
        Time series analysis plays a crucial role in various disciplines such as finance, economics, environmental science, and biomedical engineering. The primary goal of time series analysis is to understand the underlying mechanisms governing the observed data, predict future values, and identify significant patterns. Traditional approaches to time series analysis predominantly operate in the time domain, focusing on modeling trends, seasonality, and noise through techniques such as autoregressive moving average (ARMA), vector autoregression (VAR), and vector error correction models (VECM).
    
       However, these methods are often insufficient when it comes to uncovering the underlying spectral properties of the data. Analyzing time series in the frequency domain provides a complementary perspective, emphasizing periodicities and harmonic structures that are not easily observable in the time domain. To analyze the data in the frequency domain, it should be transformed by the Fourier transform, which decomposes a time series into the amplitude of harmonics for different frequencies, revealing the power distribution across those frequencies.
    
    
        The foundation of time series analysis was significantly advanced by \cite{bj70,bj94}, who introduced a systematic approach for building stationary ARMA models. Their work laid the groundwork for subsequent developments in the field, which have focused on improving model accuracy and computational efficiency \cite{brockwell91,anderson76,shumway11}. Recent advancements have also explored the integration of machine learning techniques to enhance predictive performance and adaptability to complex data structures \cite{alpaydin14,lazzeri20}.
    
    
        Further enhancing the understanding of time series, \cite{brillinger81} and \cite{priestley81} provided comprehensive treatments of frequency domain methods, highlighting their applicability to a wide range of time series data. The power spectrum, obtained through the quadratic form of the Fourier Transform of a signal, is central to these analyses, offering insights into the dominant frequencies and periodicities within the data. Frequency domain techniques have been particularly effective in fields such as signal processing and light wave, where the identification of spectral features is crucial. \cite{bloomfield00} expanded on these foundational works by elaborating on the utility of Fourier analysis for time series data, demonstrating its effectiveness in isolating and analyzing the amplitude of harmonics across different frequencies.
    
    
        Building on these spectral techniques, the ln spectrum—a logarithmic transformation of the power spectrum—improves the visualization of spectral features by expressing them as a sum of terms rather than a product, resulted from the quadratic form of the time series in the frequency domain, giving the spectra. This logarithmic approach has been extensively used in speech processing and image recognition \cite{oppenheim04} and geophysical signal analysis \cite{bogert1963,shumway82,alagon86,kakizawa98,shumway11}. The cepstrum, derived by taking the inverse Fourier Transform of the ln spectrum, further decomposes the time series variation into white noise and autocorrelation contributions. This technique has proven valuable for identifying echo patterns and deconvolving complex signals \cite{schafer1969,oppenheim1968}.
    
    
        Despite the advancements, outlier observations within time series data can significantly distort spectral estimates. Traditional periodograms can be sensitive to these outliers, leading to biased or inconsistent spectral estimates even using smooth periodograms \citep{kakizawa98,zhang92,zhang2005robust}. To mitigate these effects, robust spectral estimation techniques such as the M-periodogram and the multitaper periodogram have been developed. \cite{katkovnik98} introduced the M-periodogram as a robust method to reduce the influence of outliers. \cite{reisen20} provided an overview of robust spectral estimators and time series in the context of long memory approach, highlighting the benefits of these approaches in handling outliers and improving spectral estimation. The multitaper periodogram, as described by \cite{thomson82}, uses multiple orthogonal tapers to produce an averaged spectrum, reducing variance by diminishing the leakage caused by sidelobes and improving spectral estimation accuracy.
    
    
        To further refine classification and analysis of time series data, Linear Discriminant Analysis (LDA), introduced by \cite{fisher36}, is a statistical method used for classification and dimensionality reduction. LDA seeks to find a linear combination of features that best separates multiple classes. When applied to time series data, LDA can leverage spectral features to enhance discriminative power. Recent studies, such as the one of \cite{shumway11} have demonstrated the effectiveness of incorporating frequency domain features into LDA for improved classification of time series data. \cite{krafty16} specifically addressed the use of LDA in the presence of within-group spectral variability.
    
        Additionally, robust statistical methods have been developed to improve the performance of LDA, but there is a gap in the literature regarding robust methods in time series analysis, discrimination, and classification. \cite{huber} introduced robust estimators that can be applied in the context of time series analysis to enhance the reliability of discriminant functions. Building on these foundations, more recent works by \cite{kutz19} have further refined these techniques, making them more applicable to high-dimensional time series data.
    
    
        In this context, this paper aims to develop a robust framework for the classification of replicate groups of time series based on their spectral characteristics, using the M-periodogram as a robust estimator and accounting for within-group spectral variability. By integrating advanced spectral analysis techniques with discriminant analysis, this research offers a novel approach to time series classification. The use of robust spectral estimation method ensures that the analysis is resilient to outliers, thereby improving the reliability of the results. The proposed methodology not only improves classification accuracy but also provides deeper insights into the spectral variability inherent in time series data.
    
        Following this introduction, the paper is structured as follows. Section 2 presents the Cepstral Linear Discriminant Analysis (CLDA), detailing the theoretical foundation and methodology for applying cepstral features in discriminant analysis. Section 3 introduces the M-cepstral estimator, describing its formulation and the advantages it offers in robust spectral estimation, particularly in the presence of outliers. In Section 4, Monte Carlo simulations are conducted to evaluate the performance of the proposed methods, providing empirical evidence of their effectiveness and robustness. Section 5 applies the developed techniques to real-world data, specifically focusing on a study of neurodegenerative diseases and gait variability. This section demonstrates the practical utility of our approach in a complex biomedical context. Finally, Section 6 concludes the paper, summarizing the key findings and suggesting potential directions for future research.
    \section{Cepstral Linear Discriminat Analysis - CLDA}
        Discrimination and classification are multivariate tools used to distinguish objects based on their characteristics. According to \cite[p. 573]{wichern07}, discrimination is primarily exploratory, aiming to identify the main differences between populations. Meanwhile, \cite[p. 207]{anderson76} describes discrimination as applying algebraic or graphical rules to achieve maximal separation of time series data. On the other hand, classification involves assigning new observations to predefined populations, thereby facilitating their categorization into existing groups. The following section aims to demonstrate the extraction of time series properties for discrimination and classification, while subsequent sections will delve into the optimal procedures for effectively separating time series data.

        Let $\{X_{jkt}\}$, $t \in \mathbb{Z}$, be a family of process defined in probability space
        $\mathbb{L}^2(\Omega,\mathcal{A},\mathcal{P})$, such that $X_{jkt} = \sum_{\tau=-\infty}^{\infty} \upsilon_{jk \tau} \epsilon_{jk(t-\tau)}$ where $\{\epsilon_{jkt}\} \sim iid(0,\sigma^2)$ and ARMA coefficients with $\sum\limits_{\tau=-\infty}^{\infty}\lvert \upsilon_{\tau}\rvert < \infty$. In the above, $j=1,\cdots,J$, $k=1,\cdots,n_j$, where $J$ and $n_j$ are fixed values. The spectral density of $\{X_{jkt}\}$ is given by

        \begin{equation}\tag{2.1}\label{eq:eq2.1}
            S_{jk}(\lambda) 
            = \frac{1}{2\pi} \sum_{\tau =-\infty}^{\infty} \gamma_{jk}(\tau) cos(\lambda \tau), \quad \text{for all }\lambda\in [-\pi,\pi].
        \end{equation}
        

         \noindent where $\gamma_{jk}(\tau)$ is the covariance of the process. $S_{jk}(\lambda)$ may be interpreted as a decomposition of the variance of the process.
    
        As previously mentioned, this paper is an extension of \cite{krafty16}, that is, we address the linear discriminate analysis under time series contaminated by additive outliers or with the heavy-tailed distribution. These issues are discussed in the next sections.
        
    \subsection{Cepstra}
    
        Cepstral is a tool for investigating periodic structures in frequency spectra to extract the fundamental components of a signal, in the sense that it separates signals that have been combined in a non-additive way.   \cite{oppenheim04} gave the following example to clarify the meaning of cepstral. Suppose that a signal of a simple echo can be written as 

        \begin{equation}\tag{2.2}\label{eq:eq2.2}
            X_t=s_t+\alpha s_{(t-\tau)}.
        \end{equation}
    
        \noindent where $\alpha$ is a constant which satisfies conditions to guarantee the stability of the real process $x_t$, and $s_t$  is a stable noise process. Note that $x_t$ corresponds to MA(1) process with white noise process $s_t$ \textcolor{blue}{with lag $\tau$}.  The spectral representation of Equation \ref{eq:eq2.5} can be written as 

        \begin{equation} \tag{2.3}\label{eq:eq2.3}
            \left\lVert X_{t}(\lambda) \right\rVert^{2} = \left\lVert S_s(\lambda)\right\rVert^2 [1 + \alpha^{2}+2\alpha\cos({\lambda \tau}]. 
        \end{equation} 

        \begin{equation} \tag{2.4}\label{eq:eq2.4}
            ln\left\lVert X_{t} (\lambda)\right\rVert^{2} = ln\left\lVert S_{t}(\lambda) \right\rVert^2+ln[1 + \alpha^{2}+2\alpha\cos({\lambda \tau})].
        \end{equation}  

        \noindent where $ln = ln_e$. 
        
        From Equation \ref{eq:eq2.3} we can see that the spectral density of the echo corresponds to the product of the spectral of the noise ($S_{s}$) with the spectrum contribution of the echo. One way to see the individual contribution of echo time-variability is to apply the real ln transformation, as shown in \ref{eq:eq2.4}. From this, we  see that the $ln \left\lVert X_{t} \right\rVert^{2}$ has a waveform with the periodic component with delay $\tau$.
         
        \cite{bogert1963}, in their seminal work, introduce the  cepstrum (a noun paraphrased from the word spectrum), which is the inverse Fourier transform of the $ln \left\lVert X_{t} \right\rVert^{2}$.  The correspondent coefficients are called cepstral coefficients ($c_\ell$), $\ell = 0,1, ...$.  The cepstrum analysis is in the domain denoted as quefrency. These terminolnies were introduced by the authors where cepstral and cepstrum are anagrams of spectral and spectrum, respectively.   Apart from the mathematical elegance of the cepstrum and the cepstral coefficients, they display interesting applications in various areas of knowledge.  The cepstrum is  a tool for investigating periodic structures in frequency spectra, with main applications in human speech, music and electric power systems. 
        The periodical structures are related to noticeable echos in the signal, or to the occurrence of harmonic frequencies. Mathematically it deals with the problem of deconvolution of signals in the frequency space. The typical definition for a harmonic is ``a sinusoidal component of a periodic wave or quantity having a frequency that is an integral multiple of the fundamental frequency''. Some references refer to ``clean'' or ``pure'' power as those waveform without harmonics (\cite{Fokianos}).  See a review in \cite{cavicchioli20}.\\
        
        The cepstrum power has also been to discriminate and classification spectral densities, that is, a discriminate analysis based on the frequency and quefrency domains. \cite{Fokianos} used the methodolny 
       for testing the similarity of G spectral density functions from G-independent stationary processes. \cite{krafty16} introduces the cepstral coefficients into the Mahalanobis distance to build a discriminant function in the frequency and quefrency domains, considering variability between and within groups.

        Since we are considering the discriminant problem based on stationary zero-mean ARMA models, we derive cepstral coefficients for some examples of this class process below. 

        \begin{proposition} \label{proposition 1}
        
        Let $X_t$ be a stationary ARMA$(p,q)$ process and $\sigma^2$ is the innovation variance. Additionally, let $z=e^{-i\lambda}$, $\eta_i$ and $\zeta_r$ be the $i$th and $r$th roots of polynomials $\Theta(z)$ and $\Phi(z)$, respectively \citep[p. 155]{hamilton94}. Then:
        
        \begin{enumerate}[label=\roman*.]        
            \item The ln spectra of $X_t$ can be written as:

                \begin{equation} \tag{2.5}\label{eq:eq2.5}
                    ln S_X(\lambda) = log \left\{ \frac{\sigma_{\epsilon}^2}{2\pi} \left[ \frac{\prod_{i=1}^{q} 1+ \eta_i^2+2\eta_i cos(\lambda)} {\prod_{r=1}^{p}1+ \zeta_r^2+2\zeta_r cos(\lambda)} \right] \right \}
                \end{equation}

                $$ = ln \left\{ \frac{\sigma_{\epsilon}^2}{2\pi}\right \}+\sum_{i=1}^{q}ln\left[ 1+ \eta_i^2+2\eta_i cos(\lambda)\right] - \sum_{r=1}^{p}ln\left[ 1+ \zeta_r^2+2\zeta_rcos(\lambda) \right]$$

                $$ = ln \left\{ \frac{\sigma_{\epsilon}^2}{2\pi}\right \} + 2\left\{\sum_{i=1}^{q}\sum_{\ell=1}^{\infty} \frac{(-1)^{\ell+1}\eta_i^{\ell}}{\ell}cos(\lambda \ell)+\sum_{r=1}^{p}\sum_{\ell=1}^{\infty} \frac{\zeta_r^{\ell}}{\ell}cos(\lambda \ell)\right\}.$$

                $$ = ln \left\{ \frac{\sigma_{\epsilon}^2}{2\pi}\right \} + 2\sum_{\ell=1}^{\infty}\left\{\sum_{i=1}^{q} \frac{(-1)^{\ell+1}\eta_i^{\ell}}{\ell}cos(\lambda \ell)+\sum_{r=1}^{p} \frac{\zeta_r^{\ell}}{\ell}cos(\lambda \ell)\right\}.$$

            \item Using the ln spectra, the cepstra coefficients for $ARMA(p,q)$ can be written as:

            \begin{equation}\tag{2.6}\label{eq:eq2.6}
                c_{\ell}=    \left\{ 
                    \begin{array}{rcl}
                        \begin{matrix}    
                            ln(\frac{\sigma^2}{2\pi}) &,if & \ell = 0\\
                            &&\\
                            2\left(\sum_{i=1}^{q} \frac{(-1)^{\ell+1}\eta_i^{\ell}}{\ell}+\sum_{r=1}^{p} \frac{\zeta_r^{\ell}}{\ell}\right). &,if & \ell \geq 1.
                    \end{matrix}    
                \end{array}\right.
            \end{equation}
    
        \end{enumerate}
        \end{proposition}

        \vspace{1cm}

        In particular, Corollaries \ref{corollary 1}, \ref{corollary 2}, and \ref{corollary 3} display the ln spectra and cepstral coefficients for AR$(1)$, MA$(1)$, and ARMA$(1,1)$. The complete demonstrations are presented in Appendices 7.1 and 7.2. Also, considering the $ln = log_{e}$, where $e$ is the Neperian number. 
        
        \begin{corollary} \label{corollary 1}
        Let $X_{t} = \phi X_{(t-1)}+\epsilon_{t}$ be an stationary Gaussian $AR(1)$ process. Then:

        \begin{enumerate}[label=\roman*.]
            \item The ln spectra can be written as:

                \begin{equation} \tag{2.7}\label{eq:eq2.7}
                    ln S_X(\lambda) = ln \left\{ \frac{\sigma_{\epsilon}^2}{2\pi} \left[ 1+ \phi ^2-2\phi cos(\lambda) \right]^{-1} \right \}
                \end{equation}

                $$ = ln \left\{ \frac{\sigma_{\epsilon}^2}{2\pi}\right \} - ln\left[ 1+ \phi^2-2\phi cos(\lambda) \right]$$

                $$ = ln \left\{ \frac{\sigma_{\epsilon}^2}{2\pi}\right \} + 2\sum_{\ell=1}^{\infty} \frac{\phi^{\ell}}{\ell}cos(\lambda \ell).$$
            \item Using the ln spectra, the cepstra coefficients for $AR(1)$ can be written as:
            
            \begin{equation}\tag{2.8}\label{eq:eq2.8}
                c_{\ell}=    \left\{ 
                    \begin{array}{rcl}
                        \begin{matrix}    
                            ln(\frac{\sigma^2}{2\pi}) &,if & \ell = 0\\
                            &&\\
                            \frac{2\theta^{\ell}}{\ell} &,if & \ell \geq 1.
                    \end{matrix}    
                \end{array}\right.
            \end{equation}
        \end{enumerate}
        \end{corollary}

        Figures \ref{fig1} and \ref{fig2} display spectra (a), ln spectra (b), and cepstra (c) for the particular case of AR$(1)$ with $\phi=0.5$ and $\phi=-0.5$, respectively. In Figure \ref{fig1}, the spectra show that most of the variability of the process is explained by low frequencies, close to zero. This indicates a time series with relatively low volatility, which is confirmed by the cepstra. In panel (c), the cepstra decrease more rapidly than an exponential function toward to zero. This implies that approximately $\ell = 8$ is sufficient to explain most of the variability of the process.
      
        \begin{figure}[H]
        \centering
            \includegraphics[width=.9\linewidth, height=7cm]{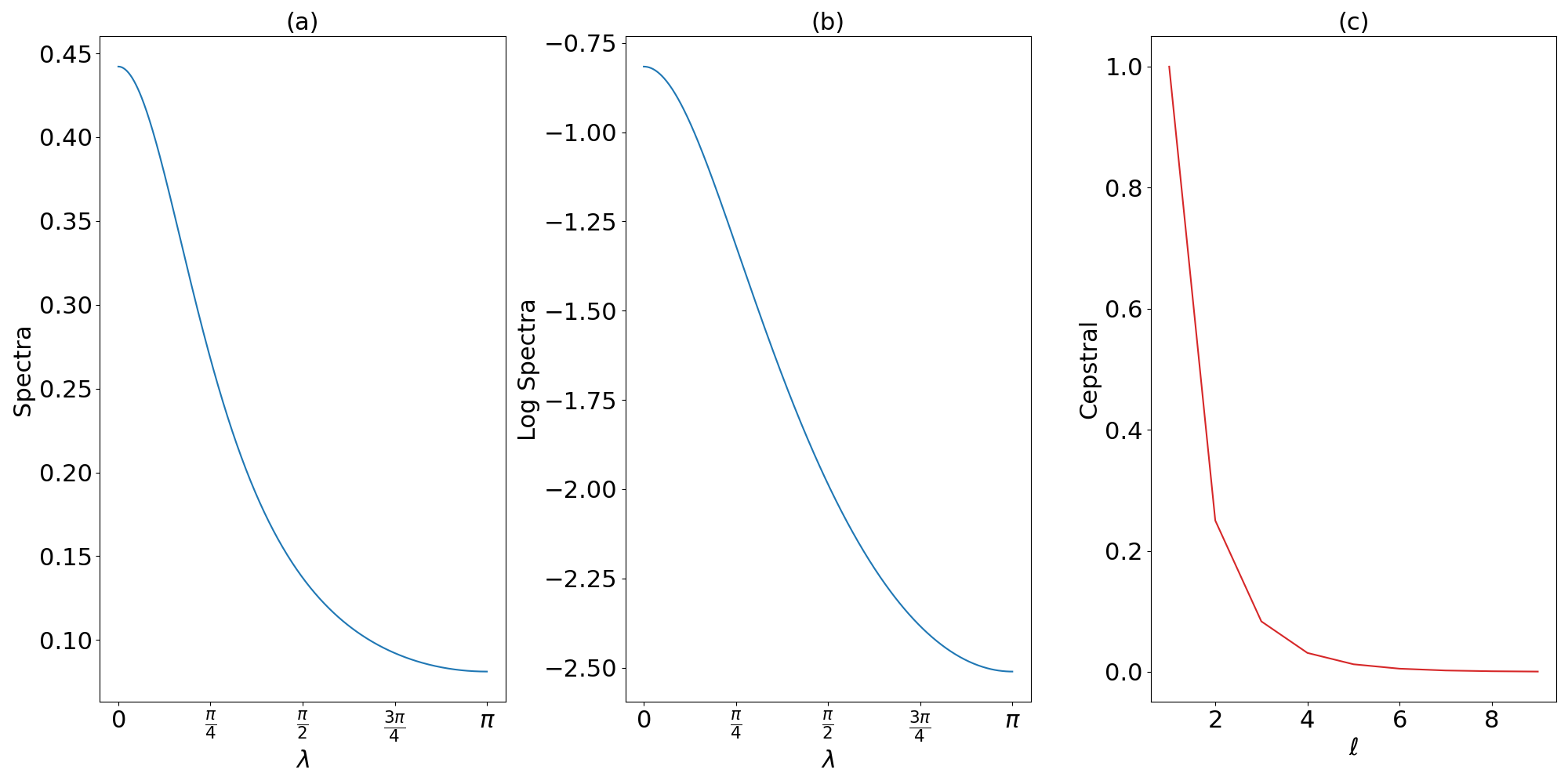}
            \caption{ $\phi = 0.5$ and $\sigma^2_{\epsilon\hspace{.1cm}} = 1$: (a) Spectra, (b) ln Spectra, (c) Cepstra}
            \label{fig1}   
        \end{figure}

        On the other hand, it is possible to see in Figure \ref{fig2} that with $\phi < 0$, the majority of the variability of the process is explained by high frequencies, close to $\pi$. As a result, the series is more volatile, and the cepstra exhibit behavior akin to a periodic function oscillating around zero. Additionally, these characteristics show that only about $\ell = 6$ is sufficient to explain the majority of the variability of the process.

        \begin{figure}[H]
        \centering
            \includegraphics[width=.9\linewidth, height=7cm]{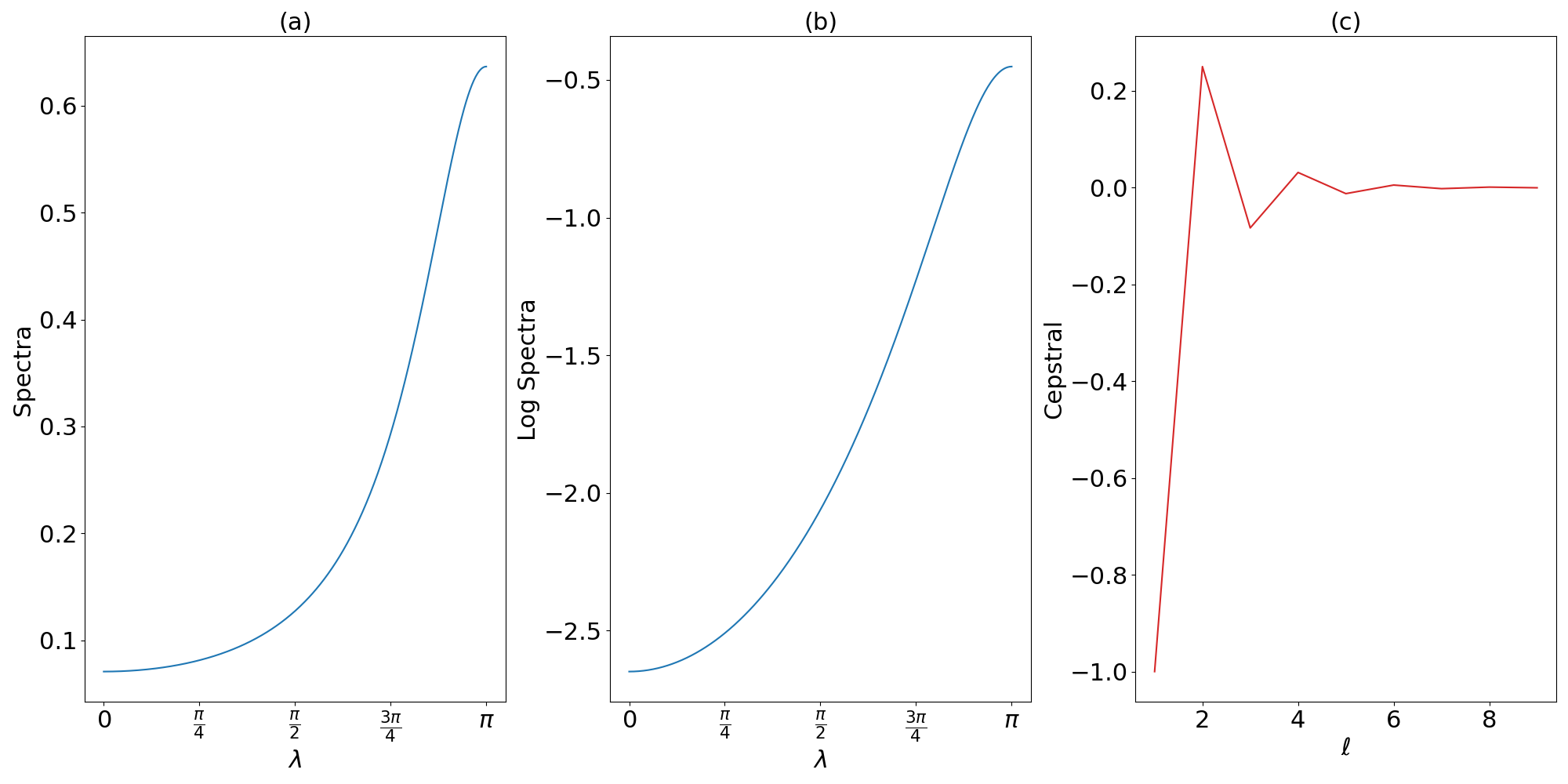}
            \caption{ $\phi = -0.5$ and $\sigma^2_{\epsilon\hspace{.1cm}} = 1$: (a) Spectra, (b) ln Spectra, (c) Cepstra}
            \label{fig2}   
        \end{figure} 

        As can be noted in Corollary \ref{corollary 2}, the behavior of cepstral coefficients for an MA(1) process differs from that of the AR(1) process presented in Corollary \ref{corollary 1} only by a change in sign for $\ell \geq 1$ and the parameter $\theta$.
        
        \begin{corollary} \label{corollary 2}
        Let $X_{t} = \epsilon_t+\theta\epsilon_{t-1}$ be an invertible Gaussian MA$(1)$ process. Then:
        
        \begin{enumerate}[label=\roman*.]
        
            \item The ln spectra can be written as:
                \begin{equation} \tag{2.9}\label{eq:eq2.9}
                    ln [S_X(\lambda)] = ln \left\{ \frac{\sigma_{\epsilon}^2}{2\pi} \left[ 1+ \theta^2+2\theta cos(\lambda) \right] \right \}
                \end{equation}

                $$ = ln \left\{ \frac{\sigma_{\epsilon}^2}{2\pi}\right \} + ln\left[ 1+ \theta^2+2\theta cos(\lambda)\right]$$

                $$ = ln \left\{ \frac{\sigma_{\epsilon}^2}{2\pi}\right \} + 2\sum_{\ell=1}^{\infty} \frac{(-1)^{\ell+1}\theta^{\ell}}{\ell}cos(\lambda \ell).$$

            \item   Using the ln spectra, it will result in cepstra coefficients for $MA(1)$, given by:

                \begin{equation}\tag{2.10}\label{eq:eq2.10}
                    c_{\ell}=    \left\{ 
                    \begin{array}{rcl}
                        \begin{matrix}    
                            ln(\frac{\sigma^2}{2\pi}) &,if & \ell = 0\\
                            &&\\
                            \frac{(-1)^{\ell+1}2\theta^{\ell}}{\ell} &,if & \ell \geq 1.
                        \end{matrix}    
                    \end{array}\right.
                \end{equation}
        \end{enumerate}
        \end{corollary}
        
        Similarly, in Figures \ref{fig3} and \ref{fig4}, panel (a) shows low and high frequencies, respectively. Thus, when $\theta > 0$, the process exhibits behavior similar to that of $\phi > 0$, while $\theta < 0$ is akin to $\phi < 0$. However, as can be noted in panel (c), the cepstra exhibit opposite behavior, indicating that low frequency in the MA$(1)$ process with $\theta = 0.5$ is associated with a periodic function in the cepstra. This is a result of the MA$(1)$ spectra having a positive sign in the cosine term, whereas the AR(1) process has a negative one.
        
        \begin{figure}[H]
        \centering
            \includegraphics[width=.9\linewidth, height=7cm]{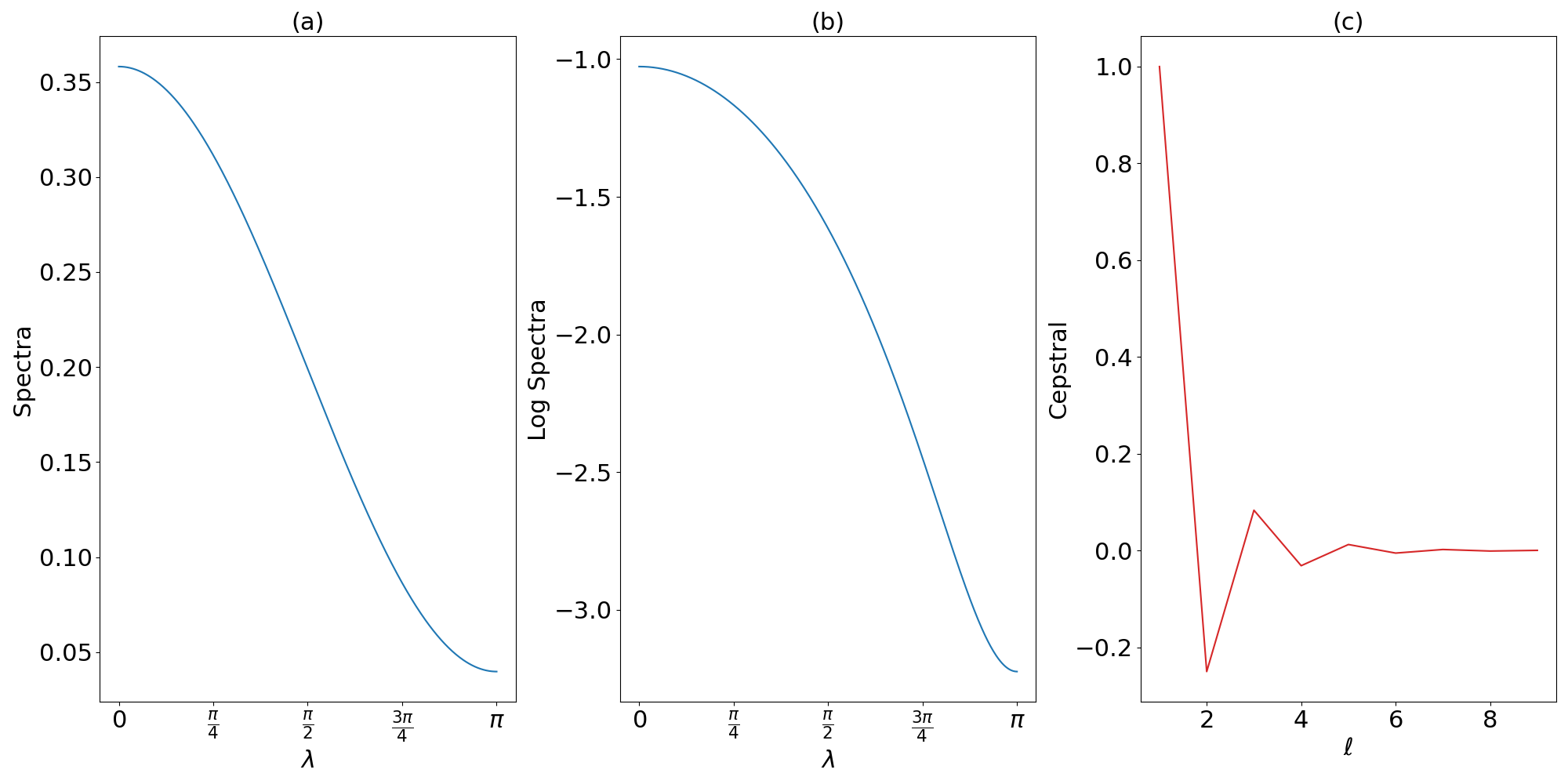}
            \caption{ $\theta = 0.5$ and $\sigma^2_{\epsilon\hspace{.1cm}} = 1$: (a) Spectra, (b) ln Spectra, (c) Cepstra}
            \label{fig3}   
        \end{figure}
        \begin{figure}[H]
        \centering
            \includegraphics[width=.9\linewidth, height=7cm]{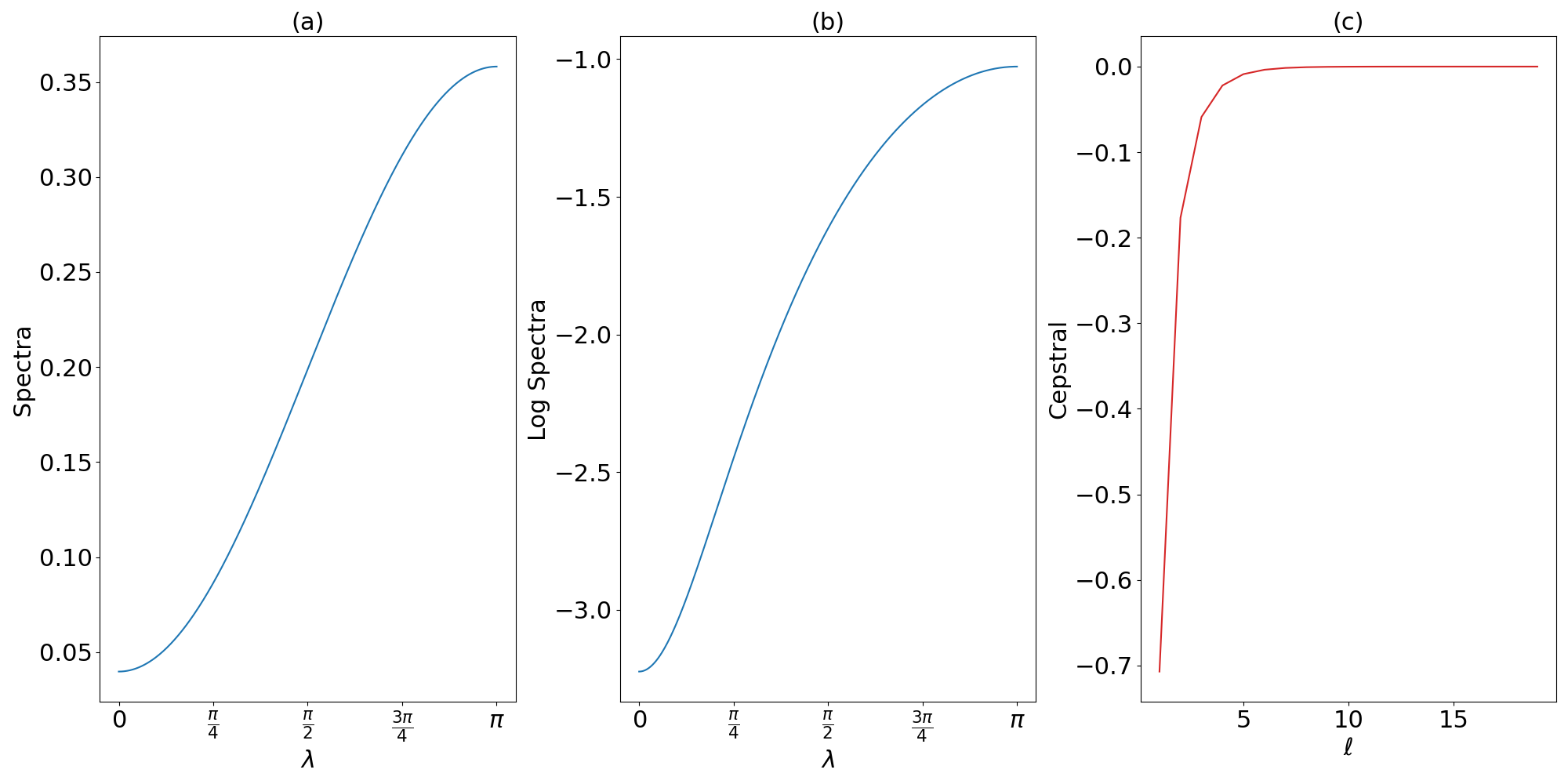}
            \caption{ $\theta = -0.5$ and $\sigma^2_{\epsilon\hspace{.1cm}} = 1$: (a) Spectra, (b) ln Spectra, (c) Cepstra}
            \label{fig4}   
        \end{figure}
        \begin{corollary} \label{corollary 3}
        Let $X_{t} = \phi X_{(t-1)}+\epsilon_{t}+\theta\epsilon_{t-1}$ be an invertible stationary Gaussian $AR(1,1)$ process. Then:
        
        \begin{enumerate}[label=\roman*.]
            
            \item The ln spectra can be written as:

            \begin{equation} \tag{2.11}\label{eq:eq2.11}
                ln S_X(\lambda) = ln \left\{ \frac{\sigma_{\epsilon}^2}{2\pi} \left[ \frac{1+ \theta^2+2\theta cos(\lambda)} {1+ \phi^2-2\phi cos(\lambda)} \right] \right \}
            \end{equation}

            $$ = ln \left\{ \frac{\sigma_{\epsilon}^2}{2\pi}\right \}+ln\left[ 1+ \theta^2+2\theta cos(\lambda)\right] - ln\left[ 1+ \phi^2-2\phi cos(\lambda) \right]$$

            $$ = ln \left\{ \frac{\sigma_{\epsilon}^2}{2\pi}\right \} + 2\left\{\sum_{\ell=1}^{\infty} \frac{(-1)^{\ell+1}\theta^{\ell}}{\ell}cos(\lambda \ell)+\sum_{\ell=1}^{\infty} \frac{\phi^{\ell}}{\ell}cos(\lambda \ell)\right\}.$$

            \item   Using the ln spectra, it will result in cepstra coefficients for $ARMA(1,1)$, given by:
                \begin{equation}\tag{2.12}\label{eq:eq2.12}
                    c_{\ell}=    \left\{ 
                    \begin{array}{rcl}
                        \begin{matrix}    
                            ln(\frac{\sigma^2}{2\pi}) &,if & \ell = 0\\
                            &&\\
                            \frac{2\theta^{\ell}}{\ell}+ \frac{(-1)^{\ell+1}2\theta^{\ell}}{\ell} &,if & \ell \geq 1.
                        \end{matrix}    
                    \end{array}\right.
                \end{equation}
        \end{enumerate}
        \end{corollary}

        \begin{corollary} \label{corollary 4} Under the Assumption of a stationary ARMA model, the cepstral coefficients are of order $c_{\ell} = \mathcal{O}(\ell^{-3/2})  = \operatorname{o}(\ell^{-1/2}) \quad \text{as} \quad \ell \to \infty$.
            
        \end{corollary}

        \begin{remark} \label{remark 1}
            Note that the cepstral coefficient decays at the rate $c_\ell$ $\sim$  $\vartheta|\frac{\delta ^{\ell}}{\ell}|$, $-\infty < \ell < \infty$, where $\vartheta$ is a constant and $|\delta| <1$, while the autocorrelation function of the standard Box-Jenkins ARMA models decays approximately in the exponential form. For example, the ACF of the AR(1) model is $\phi^l$. The mathematical different behaviour between the cepstrum coefficients and the ACF values clearly show that the former displays most of the time-variation of the process with a smaller lag than the ACF function, indicating that most information is contained in lower-order cepstral coefficients.
        \end{remark}

        Figures \ref{fig1}, \ref{fig2}, \ref{fig3} and \ref{fig4} display the behaviour of spectra, ln spectra and cepstra of MA(1) and AR(1) models, respectively.  In both cases, we see the frequencies near zero, which contribute most to the process variance. The cepstrum shows that the initial lags contain all the information necessary to extract the features of the processes. For positive $\phi$, the series is positively correlated the spectrum and cepstrum are dominated by low frequency and small lags, respectively,   which means that the series is relatively smoothed. When $\phi <0$, the series is negatively correlated, thus the spectrum and cepstrum are dominated by high-frequencies and lags, respectively. This means that the data is more ragged series. 

        \subsection{Cepstral Linear Discriminant Analysis}

                Let $\pmb{c^{(\ell)}_{jk\ell}} = [c_{j1\ell} \quad c_{j2\ell} \quad \cdots \quad c_{jn_j\ell}]^{T}$ be the cepstral vector for a fixed $\ell$ of the $j$th population $\Pi_j$, which is centered in the mean  $\mathbb{E}(c_{jk\ell}  \mid \Pi=j ) =$ $\mu_{j,l}$ and  
                $\pmb{c^{(k)}_{jk\ell}} = \left[c_{jk0} \quad c_{jk1} \quad \cdots \quad c_{jk\ell}\right]^{T}$ be the cepstral vector for a fixed replication $k$. The centroids (the mean vector) are defined as $\pmb{\mu_{j\ell}^{(k)}} = [\mu_{j0} \quad \mu_{j1} \quad \cdots \quad \mu_{j\ell}]^{T}$ , $\pmb{\mu_{j\ell}^{(\ell)}} = [\mu_{1\ell} \quad \mu_{2\ell} \quad \cdots \quad \mu_{J1\ell}]^{T}$ being for a given population, the vector mean across replicates and the vector mean across population, given a fixed cepstra $\ell$, respectively.
    
               Therefore, let ($\mu_{j\ell}$) be the within-mean by cepstra, ($\pmb{\mu}$) be the overall mean centroid, ($\Omega_W$) be the within-variance and ($\Omega_B$) be the between-variance. Moreover, the operator $\left\langle\,\cdot,\cdot\right\rangle$ is the inner product in the Euclidean space $\mathbb{R}^{n}$. Then, defined as follows:
    
            \begin{enumerate}[label=\roman*.] 
                \item \textbf{Within-class mean by cepstra:}
                    \begin{equation}\tag{2.13}\label{eq:eq2.13}
                        \mu_{j\ell} = \mathbb{E}(\mu_{jk\ell}) = \mathbb{E}(\mathbf{c}^{(\ell)}_{jk\ell} \mid \Pi=j).
                    \end{equation}
            
                \item \textbf{Overall mean centroid:}
                    \begin{equation}\tag{2.14}\label{eq:eq2.14}
                        \mathbf{\mu} = \mathbb{E}(\mathbf{c}^{(\ell)}_{jk\ell}) = \mathbb{E}[\mathbb{E}(\mathbf{c}^{(\ell)}_{jk\ell} \mid \Pi=j )] = \left[\left\langle \mu_{j\ell}^{(0)}, \mathbf{f} \right\rangle \quad \left\langle \mu_{j\ell}^{(1)}, \mathbf{f} \right\rangle \cdots \left\langle \mu_{j\ell}^{(L-1)}, \mathbf{f} \right\rangle \right]^T.
                    \end{equation}

                \item \textbf{Between-class variance:}
                    \begin{equation}\tag{2.15}\label{eq:eq2.15}
                        \Omega_B = \text{Var}[\mathbb{E}(\mathbf{c}_{j\ell} \mid \Pi=j)] = \left\lVert \mu_j - \mu \right\rVert^2 \left\langle \mathbf{\iota_{n}}, \mathbf{\iota_{n}} \right\rangle^{-1} = d^2(\mu_j, \mu) \left\langle \mathbf{\iota_{n}}, \mathbf{\iota_{n}} \right\rangle^{-1}.
                    \end{equation}

                \item \textbf{Within-class variance:}
                    \begin{equation}\tag{2.16}\label{eq:eq2.16}
                        \Omega_W = \mathbb{E}[\text{Var}(\mathbf{c}^{(k)}_{jk} \mid \Pi=j)] = \mathbb{E}\left[ (\mathbf{c}^{(k)}_{jk} - \mu_j)(\mathbf{c}^{(k)}_{jk} - \mu_j)^{T} \mid \Pi=j) \right].
                    \end{equation}

            \end{enumerate}

            Since $\Omega_W$ and $\Omega_B$ are  positive definite matrices,  the discriminant coefficients are derived from the following generalization of the maximization of the quadratic form on the unit sphere, usually called Generalized Rayleigh Quotient \citep[p. 1192]{shin08}
    
            \begin{equation}\tag{2.17}\label{eq:eq2.17}
                \pmb{p_1} = \underset{\pmb{p} \in \mathbb{R}^{\mathbb{L}}}{\operatorname{argmax}}
                    \left\{
                        \frac{Var[\mathbb{E}(\pmb{p}^{T}\pmb{c_{jk}}|\Pi{j}=j)]}{\mathbb{E}[Var(\pmb{p}^{T}\pmb{c_{jk}}|\Pi{j}=j)]}
                    \right\} 
                    = \underset{\pmb{p} \in \mathbb{R}^{\mathbb{L}}}{\operatorname{argmax}}
                    \left\{
                    \frac{\pmb{p}^T\Omega_B \pmb{p}}{\pmb{p}^T\Omega_w \pmb{p}} 
                    \right\}.
            \end{equation}
    
            or equivalently
    
            \begin{equation}\tag{2.18}\label{eq:eq2.18}
            \begin{aligned}
            \max_{\pmb{p}\in \R^{\mathbb{L}}} \quad & \pmb{p}^{T}\Omega_{B}\pmb{p}\\
            \text{subject to} \quad &
                    \pmb{p}^{T}\Omega_{w}\pmb{p} = 1. & \\ 
                \end{aligned}
            \end{equation}
            
            which provide a given $\pmb{p_1} = [p_{10},\cdots,p_{1(L-1)}]^T \neq 0$.
            
            The optimization problem can be achieved     sequentially by Lagrange multiplier
    
            \begin{equation}\tag{2.19}\label{eq:eq2.19}
                \mathcal{L}(\pmb{p}) =  \pmb{p}^{T}\Omega_{B}\pmb{p} - \lambda(\pmb{p}^{T}\Omega_{w}\pmb{p}-1).
            \end{equation}
    
            The first order condition is given by the gradient vector equals to zero, i.e., $\nabla \mathbb{L}(\pmb{p}) = 0$, then
    
            \begin{equation}\tag{2.20}\label{eq:eq2.20}
                2\Omega_{B}\pmb{p} - 2\lambda\Omega_{w}\pmb{p} = 0 \Longleftrightarrow{} (\Omega_{B} - \lambda\Omega_{w})\pmb{p}=0  \Longleftrightarrow{}(\Omega_{w}^{-1}\Omega_{B} - \lambda I)\pmb{p}=0.
            \end{equation}
    
            \noindent where $I$ is the identity matrix.
    
            The solution of Equation \ref{eq:eq2.18} (\ref{eq:eq2.20}) gives the generalized eigenvector  $\pmb{p}$  see, for example, \cite[p. 497]{golub}. Based on the Extended Cauchy-Schwartz Inequality,  it is easy to show that the first orthonormal vector, say $\pmb{p_1}$, is given from the largest eigenvalue $\lambda_1$ of the matrix $\Omega_{w}^{-1}\Omega_{B}$. Then, it provides the rate of the two quadratic forms where the numerator takes in account the variability matrix between groups ($\Omega_{B}$) and the denominator corresponds to variability within group ($\Omega_{w}$). That is, the above solutions can also be derived from     
            
            \begin{equation}\tag{2.21}\label{eq:eq2.21}
                \pmb{p^{*}_1} = \underset{\pmb{p} \in \mathbb{R}^{\mathbb{L}}}{\operatorname{argmax}}
                    \left\{
                        \frac{\pmb{p}^T \Omega_w^{-1} \Omega_B \pmb{p}}{\pmb{p}^T \pmb{p}}
                    \right\}.
            \end{equation}
    
            \noindent which leads to the following proposition.
    
            \begin{proposition} \label{proposition 2} Let $A$ and $B$ both be symmetric positive definite $(L \times L)$ and $\pmb{x}$ and $\pmb{y}$ are  orthonormal $(L\times1)$ vectors, then $\pmb{x_1} = \pmb{y_1}$, where 
    
            \begin{equation*}
                \pmb{x_1} = \underset{\pmb{x} \in \mathbb{R}^{\mathbb{L}}}{\operatorname{argmax}}
                    \left\{
                    \frac{\pmb{x}^T A \pmb{x}}{\pmb{x}^T B \pmb{x}} 
                    \right\}. 
                \end{equation*}
    
                and
    
                \begin{equation*}
                \pmb{y_1} = \underset{\pmb{y} \in \mathbb{R}^{\mathbb{L}}}{\operatorname{argmax}}
                    \left\{
                    \frac{\pmb{y}^T B^{-1}A \pmb{y}}{\pmb{y}^T\pmb{y}} 
                    \right\}. 
                \end{equation*}
                
            \end{proposition}

            In the context of Equation \ref{eq:eq2.19}, the optimization process is a component principal analysis of $\Omega_w^{-1}\Omega_B$. Therefore, we seek for the discriminant function, which is a linear combination of the cepstral such that the conditional variance-between is as large as possible, relative to the conditional variance-within \cite[p. 431]{wichern07}. Then, the first linear combination (discriminants) gives the maximum relative conditional variance.
    
            \begin{remark}Note that the denominator of the Equation \ref{eq:eq2.19} $\pmb{p^T}\pmb{p}$ is actually a condition that $\pmb{p^T}\pmb{p} = 1$,i.e., unitary norm. It is a necessary condition to the vector $p$ that is going to be found for non increasing relative conditional variance of the discriminant function $Var[\mathbb{E}(\pmb{d_{jk}}|\Pi{j}=j)] = Var[\mathbb{E}(\pmb{p_1^*}^{T}c_{jk}|\Pi{j}=j)] = \pmb{p_1^*}^{T} \Omega_w^{-1} \Omega_B \pmb{p_1^*}$. 
            \end{remark}
    
            Note that $\lambda_1$ generated by the above results corresponds to the largest variability of the hyperellipsoids of $\Omega_{w}^{-1}\Omega_{B}$, in which  is related to the axis generated by $\pmb{p_1}$. In the sequence, $\lambda_2> \quad \lambda_3 \quad >\cdots > \quad \lambda_q$ are the correspondent axis generated by the orthonormal basis space $\pmb{p_2},..., \pmb{p_q}$  , where $q = min\{J-1,rank(A),rank(B)\}$ is the dimension of the subspace. Thus, the set of vectors $\pmb{p}_{1}, \pmb{p}_{2}, \cdots, \pmb{p}_{q}$ will be a basis of a subspace that best separates cepstra centroids from each other related to within-variance. Then, all replicates cepstra should be projected into the new subspace by a linear combination, which gives the discriminant functions (coordinates)
    
            \begin{equation}\tag{2.22}\label{eq:eq2.22}
                d_{jkq} = \left\langle\,\pmb{p}_{\ell},\pmb{c^{(k)}_{jk\ell}} \right\rangle =  \sum_{\ell=0}^{\infty} p_{q\ell}c_{jk\ell}.
            \end{equation}
    
            \subsubsection{Decision Rules}
    
            To establish the discriminant  boundaries of the hyperelliposids, that is, the decision area, this is computed using the loss (in making a decision) function based on square error loss given by
    
            \begin{equation}\tag{2.23}\label{eq:eq2.23}
                F(d_{*1},\cdots,d_{*q};\mu_j) = \left[\left\lVert \pmb{d}-\pmb{\mu_j} \right\rVert^2 - 2\ln(f_j) \right]=\left[\sum_{s=1}^{q} (d_{*s}-\mu_{js})^2 - 2\ln(f_j) \right]. 
            \end{equation}
            
            \noindent where $\pmb{d}{} = \left(d{*1}, \cdots, d_{*q}\right)^{T}$ stands as the discriminant vector for the unknown population time series, and $2\ln(f_j)$ is a correction factor for the number of time series in  population $j=1,...,J$.
    
            
            Once each replicated time series is classified, the classification rate can be conducted by each population and can be represented in a confusion matrix, which is a tool used to evaluate the performance of a classification model. It presents a summary of the prediction results on a classification problem. The matrix itself is a Table with two dimensions: "population" and "predicted population," and it is used to visualize the performance of the algorithm.
    
            In the context of classifying time series into $J$ populations using cepstral coefficients from stationary ARMA models, the confusion matrix provides insights into how well your Linear Discriminant Analysis (LDA) model is performing by comparing the actual population labels of the time series against the predicted labels.
            
            \begin{table}[H]
                \centering
                \caption{Normalized Confusion Matrix for $J$ Populations}
                \label{tab:1}
                \begin{adjustbox}{totalheight={3.9cm}}
                    \begin{tabular}{c|ccccc} 
                    \hline
                    \diagbox{\textbf{Predicted Population}}{\textbf{Actual Population}} & \textbf{1} & \textbf{2} & \textbf{3} & \textbf{...} & \textbf{J} \\ 
                    \hline
                    \textbf{1} & $\rho_{11}$ & $\rho_{12}$ & $\rho_{13}$ & ... & $\rho_{1J}$ \\
                    \textbf{2} & $\rho_{21}$ & $\rho_{22}$ & $\rho_{23}$ & ... & $\rho_{2J}$ \\
                    \textbf{3} & $\rho_{31}$ & $\rho_{32}$ & $\rho_{33}$ & ... & $\rho_{3J}$ \\
                    \textbf{...} & ... & ... & ... & ... & ... \\
                    \textbf{J} & $\rho_{J1}$ & $\rho_{J2}$ & $\rho_{J3}$ & ... & $\rho_{JJ}$ \\
                    \hline
                    total & 1 & 1 & 1 & ... & 1 \\
                    \hline\hline
                    \end{tabular}
                \end{adjustbox}
            \end{table}
            In this matrix, $\rho_{ij}$ represents the proportion of time series from the actual population $i$ that have been classified as belonging to the predicted population $j$. Each cell $\rho_{ij}$ is the ratio of the number of time series correctly or incorrectly classified into population $j$ to the total number of time series in the actual population $i$.
            
            The sum of $\rho_{ij}$ for each column (i.e., for each predicted population) equals 1 because all time series from the actual populations must be classified into one of the predicted populations, covering all possible classification outcomes. This reflects the total distribution of possible classifications and ensures that all predictions are accounted for, allowing for an accurate analysis of the classification model's performance.
            
            When using LDA to classify time series data into $J$ populations based on cepstral coefficients from stationary ARMA models, each time series is represented by its cepstral coefficients. The LDA is trained on the cepstral coefficients to learn the discriminant functions that separate the populations. The trained LDA model is used to predict the population for each time series in the test set. The confusion matrix is generated by comparing the actual population labels of the time series in the test set with the predicted labels.
            
            The steps to generate and use a confusion matrix are as follows: First, train the LDA model using the training set of time series data with known population labels. Then, apply the trained LDA model to the test set to obtain predicted population labels. Compare the actual labels with the predicted labels to populate the confusion matrix. Finally, identify which populations are often confused with each other and adjust the model or preprocessing steps to improve classification performance based on the insights obtained.
            
            The benefits of using a confusion matrix include detailed error analysis, which helps identify specific classes that are being misclassified and understand the nature of the misclassifications. It also aids in model improvement by providing insights into how to adjust the model or preprocessing steps to enhance classification performance. By using the confusion matrix in our classification analysis with LDA, we can gain a deep understanding of how well our model distinguishes between the populations and identify areas for improvement.
    \section{Estimation Methods}
    
            Let $\Pi_j$ be the  $j$th population under the study and $\{X_{jk0},\cdots, X_{jk(N-1)}\}$ the $k$th realization of the process $\{X_{jkt}\}$ in the population $j=1,\cdots,J$. The total number of realizations corresponds to $n = \sum_{j=1}^J n_j$. The proportion of time series in each population is $f_j = n_j/n = \left\langle\,\pmb{\iota_n},\pmb{\iota_n} \right\rangle^{-1} \left\langle\,\pmb{\iota_{n_j}},\pmb{\iota_{n_j}} \right\rangle = \left\langle\,\pmb{\iota_n},\pmb{\iota_n} \right\rangle^{-1} \left\langle\,\pmb{\iota_{n_j}},\pmb{\iota_{n_j}} \right\rangle$, whose vector is given by $\pmb{f} = [f_1 \quad f_2 \quad \cdots \quad f_J]^{T}$, the vectors $\pmb{\iota_{n}}$ and $\pmb{\iota_{n_j}}$ are vectors  of ones with size $(n\times1)$ and $(n_j\times1)$, respectively. 
    
            As stated previously, the cepstral estimated coefficients are the result of the inverse Fourier transform of the logarithm of the spectra, i.e., the inverse Fourier transform of \( \ln \left( S_X(\lambda) \right ) \) for a fixed \( \lambda \). The logarithm is important to decompose the product caused by the convolution to generate the spectra, separating the contribution of the white noise and the autocorrelations to the total variance of the process into sum terms. Since all processes are considered to be stationary, the Fourier transform needed is just using the cosine series \( \cos(\lambda \ell) \), \( \ell = 1, 2, \ldots \), generating the general cepstral coefficients given by 

            \begin{equation}\tag{3.1}\label{eq:eq3.1}
                    \hat{c}_{jk\ell}=    \left\{ 
                    \begin{array}{rcl}
                        \begin{matrix}    
                            N^{-1}\sum_{m=0}^{N-1}ln \left(\hat{S}_{jk}(\lambda_m) \right) &,if & \ell = 0\\
                            &&\\
                            N^{-1}\sum_{m=0}^{N-1} ln \left(\hat{S}_{jk}(\lambda_m)\right)cos \left( \lambda_m \ell \right) &,if & \ell \geq 1.
                        \end{matrix}    
                    \end{array}\right.
                \end{equation}

            \noindent where $\hat{S}_{jk}(\lambda_m)$ is an estimator of the spectra for the $k$th replicates of the population $j$ and $\lambda_m = 2\pi m /N$, $m \in \mathbb{Z}$. Any estimator of spectra is appropriate to estimate cepstral coefficients described in Equation \ref{eq:eq2.6}, once for a fixed $\ell$, $Var \left( \hat{c}_{jk\ell} - c_{jk\ell} \right) \rightarrow 0$ when $N \rightarrow \infty$.
            
            Based on Corollary \ref{corollary 4} and Remark \ref{remark 1}, these coefficients decay fast for ARMA$(p,q)$ process. Therefore, the $ln \left( S_X(\lambda) \right )$ can be approximated by using  finite numbers fixed of the coefficients, that is, hereafter $\ell= 0, 1, 2, ...., L-1$, where $L$ is an integer positive value such that  $L < N, n$.
            
            Once the cepstral coefficients are estimated, for each $j$th population $\Pi_j$, the cepstral vectors can be estimated by  $\pmb{\hat{c}^{(\ell)}_{jk\ell}} = [\hat{c}_{j1\ell} \quad \hat{c}_{j2\ell} \quad \cdots \quad \hat{c}_{jn_j(L-1)}]^{T}$ for each fixed $\ell = 0,...,(L-1)$ vector centered in the mean  $\mathbb{E}(\hat{c}_{jk\ell}  \mid \Pi=j ) =$ $\mu_{j,l}$ and  
            $\pmb{\hat{c}^{(k)}_{jk\ell}} = \left[\hat{c}_{jk0} \quad \hat{c}_{jk1} \quad \cdots \quad \hat{c}_{jkL-1}\right]^{T}$ cepstral vector for a fixed replication $k$. The centroids (mean vector) are defined as $\pmb{\hat{\mu}_{j\ell}^{(k)}} = [\hat{\mu}_{j0} \quad \hat{\mu}_{j1} \quad \cdots \quad \hat{\mu}_{j(L-1)}]^{T}$ , $\pmb{\hat{\mu}_{j\ell}^{(\ell)}} = [\hat{\mu}_{1\ell} \quad \hat{\mu}_{2\ell} \quad \cdots \quad \hat{\mu}_{J1\ell}]^{T}$ being for a given population, the vector mean across replicates and the vector mean across population, given a fixed cepstra $\ell$, respectively. Then, all position and scale measures for estimated cepstral can be defined as:

            \begin{enumerate}[label=\roman*.] 
                \item \textbf{Within-class mean by cepstra:}

                \begin{equation}\tag{3.2}\label{eq:eq3.2}
                    \hat{\mu}_{j\ell} = \frac{1}{n_j} \sum_{k=1}^{n_j} \hat{c}^{(\ell)}_{jk\ell} = \left\langle \mathbf{\iota_{n_j}}, \hat{\mathbf{c}}^{(\ell)}_{jk\ell} \right\rangle \left\langle \mathbf{\iota_{n_j}}, \mathbf{\iota_{n_j}} \right\rangle^{-1}.
                \end{equation}
                
                \item \textbf{Overall mean centroid:}
                    \begin{equation}\tag{3.3}\label{eq:eq3.3}
                        \hat{\mathbf{\mu}} = \left[\left\langle \mu_{j\ell}^{(0)}, \mathbf{f} \right\rangle \quad \left\langle \mu_{j\ell}^{(1)}, \mathbf{f} \right\rangle \cdots \left\langle \mu_{j\ell}^{(L-1)}, \mathbf{f} \right\rangle
                    \right]^T.
                    \end{equation}
            
                \item \textbf{Between-class variance:}
                    \begin{equation}\tag{3.4}\label{eq:eq3.4}
                        \hat{\Omega}_B = \left\lVert \hat{\mu}_j - \hat{\mu} \right\rVert^2 \left\langle \mathbf{\iota_{n}}, \mathbf{\iota_{n}} \right\rangle^{-1} = d^2(\hat{\mu}_j, \hat{\mu}) \left\langle \mathbf{\iota_{n}}, \mathbf{\iota_{n}} \right\rangle^{-1}.
                    \end{equation}
            
                \item \textbf{Within-class variance:}
                \begin{equation}\tag{3.5}\label{eq:eq3.5}
                        \hat{\Omega}_W = (n_j)^{-1} \sum_{k=1}^{n_j}\left[ (\hat{\mathbf{c}}^{(k)}_{jk} - \hat{\mu}_j)(\hat{\mathbf{c}}^{(k)}_{jk} - \hat{\mu}_j)^{T} \right].
                    \end{equation}
            \end{enumerate}

        Once the variance of the cepstral coefficients given the population is calculated, it should have a minimum superior cote. Let it be defined as follows

        \begin{equation}\tag{3.6}\label{eq:eq3.6}
            \hat{\sigma}_L = \sup \left\{ E[(\hat{c}_{jk(L-1)}-\hat{\mu}_j)(\hat{c}_{jk(L-1)}-\hat{\mu}_j)^T] \right\}
        \end{equation}

        Once the cepstral coefficients are calculated, the discriminant function can also be estimated as 

        \begin{equation}\tag{3.7}\label{eq:eq3.7}
            \hat{d}_{jks} = \left\langle\,\pmb{\hat{p}}_{\ell},\pmb{\hat{c}^{(k)}_{jk\ell}} \right\rangle =  \sum_{\ell=0}^{L-1} \hat{p}_{s\ell}\hat{c}_{jk\ell}.
            \end{equation}

        Then, the decision rule in Equation Equation \ref{eq:eq2.23} can be estimated $\hat{F}(\hat{d}_{*1},\cdots,\hat{d}_{*q};\hat{\mu}_j)$, as well as the confusion matrix described in Table \ref{tab:1}.
            
        The estimation methods based on the periodograms will be discussed in next sections.

        \subsection{Cepstral Estimation based on Classical Periodogram}

            A natural estimator of the spectral density \ref{eq:eq2.1} is the periodogram function 
    
            \begin{equation}\tag{3.8}\label{eq:eq3.8}
                I_{jk} (\lambda_{m}) = \frac{1}{2\pi N} \left\lVert\sum_{t=0}^{N-1} X_{jkt}e^{-i \lambda_m t}\right\rVert^2.
            \end{equation}

            \textcolor{black}{\begin{equation*}
                =\frac{1}{2\pi}\sum_{\nu=-(N-1)}^{N-1} \hat{\gamma}_{jk}(\tau)e^{-i\lambda_{m} \nu}
            \end{equation*}}

            \noindent \textcolor{black}{where}

            \textcolor{black}{\begin{equation*}\tag{3.9}\label{eq:eq3.9}
                \hat{\gamma}_{jk}(\tau) = \frac{1}{N} \sum_{t=1}^{N-|\nu|} X_{jkt}X_{jk (t-|\nu|)}
            \end{equation*}}

            \noindent \textcolor{black}{is the sample autocovariance function}, $\lambda_{m} = 2\pi m/N$, $m = 1, \ldots, \lfloor N/2 \rfloor$, are the Fourier frequencies $\left\lVert \cdot \right\rVert$ means Euclidean norm. As well known, the periodogram is an asymptotically unbiased estimator for the spectral density, however, its asymptotic variance is of order $\mathcal{O}$ $(1)$, that is, the periodogram is not a consistent estimator of the spectral density.
    
            As shown in \cite[p. 206]{percival}, the bias in periodogram is due to the sidelobes of Ferjér's kernel, which leads to leakage of the information of one frequency to another. To reduce the sidelobes, and consequent bias, \cite{percival} suggest the use of tapering the data as originally proposed by \cite{thomson82}. The idea is form another data tapered $h_tX_{jkt}$, where $h_t$ is called a taper funcion. The main idea is construct another data base in order to $X_{jkt}=0$ for $t$ near to $0$ and $N$. It will reduce the leakage and the bias. Then, a new periodogram, called \textit{modified periodogram} can be achieved for the new data. The taper function is, actually, a window function, but different from classical approach in the linterature, it is computationally more efficient and has high resolution if it is used the sine tepers, defined as
    
            \begin{equation}\tag{3.10}\label{eq:eq3.10}
                h_{rt}= \left(\frac{2}{N+1}\right)^{1/2}sin\left(\pi t \frac{r}{N+1}\right).
            \end{equation}
        
            \noindent where $r =1,\cdots,R$ and $h_{rt}$ is the $r$th weight taper of $X_{jkt}$.  In this case, the $r$th taper will be multiplied by the time series realization and a periodogram is obtained. Then, it should be repeated $R$ times and an average periodogram can be calculated by

            \begin{equation}\tag{3.11}\label{eq:eq3.11}
                I_{jk}^{r} (\lambda_{m}) =\frac{1}{2\pi N}\left\lVert\sum_{t=1}^N h_{rt}X_{jkt}e^{-i \lambda_{m} t}\right\rVert^2
            \end{equation}
    
            \begin{equation}\tag{3.12}\label{eq:eq3.12}
                I^{R}_{jk}(\lambda_m) = \frac{1}{R} \sum_{r=1}^R I^{r}_{jk}(\lambda_m).
            \end{equation}
   
           \begin{remark}\label{remark 3}
              Note that Equation \ref{eq:eq3.10} is a type of smoothed periodogram, that is, it preserves the properties of being asymptotically unbiased and consistent estimator of the spectral density. In addition, for each fixed $\lambda_m$, $j,k$ and $r=1,...,R$,  the variables $I_{jk}^r(\lambda_m)$ are independent. 
           \end{remark}
            
         \begin{remark} \label{remark 4}
            Under the Assumption that $R$ is a function of $N$ ($R_N$) such $\frac{R_N}{N} \rightarrow 0$ as $R_N$ and N go to $\infty$, $I^{R}_{jk}(\lambda_m)$ also corresponds to a window spectral estimator, i.e. it is asymptotically unbiased and consistent estimator see, for example, \cite[p. 358]{brockwell91}.
         \end{remark}

          Now some additional Assumption made to guarantee the consistency property of the cepstral estimator. 

            \textit{Assumptions}

                \begin{assumptions} \label{A.1}
                    $\sum_{\tau=1}^{\infty} \lvert \upsilon_{\tau}\rvert \lvert \tau^{1/2} \rvert < \infty$ and $E \lvert \epsilon_{jkt} \rvert^4 < \infty$.
                \end{assumptions}

                \begin{assumptions} \label{A.2}

                    If $\hat{\sigma}_L < \infty$ $\Rightarrow$ $n^{-1/2}\hat{\sigma}_L^{-1} \rightarrow 0$, $N^{-1/2}\hat{\sigma}_L^{-1} \rightarrow 0$ and $n^{-1/2}L \rightarrow 0$ when $n,N,L \rightarrow \infty$.
                \end{assumptions}\

                Assumption A1 is the standard Assumption made for a smoothed ( window) periodogram estimator (\cite[Theorem 10.4.1]{brockwell91}.  It is also possible to prove that $\hat{f_j}$ that Assumption A2 converge in probability of a binomial distribution  with probability $p$ estimated by $n_j/n$. Assumption 3 is the regularity condition of the $\Omega_w$ to be non singular.

            \textit{Asymptotic Results}
            
            \begin{proposition}
                In \ref{eq:eq3.1}, let $\hat{S}_{jk}(\lambda_m)=$ $I^{R_N}_{jk}(\lambda_m)$ (\ref{eq:eq3.11}), where $R=R_N$ is defined in Remark \ref{remark 4}. Thus,  $\hat{c}_{jk\ell} \xrightarrow{p} c_{jk\ell}$ when $\frac{R_N}{N} \rightarrow 0$ and $N,R_N \rightarrow \infty$.
            \end{proposition}

            \begin{proof}
                The proof is straightforward obtained using Theorem 9.4.2 in \cite[p. 541]{anderson76}.
            \end{proof}

            \begin{lemma} For $n\rightarrow \infty$ and for a fixed $j=1,\cdots,J$ $\hat{f}_j \xrightarrow{p} f_j$.
            \end{lemma}

            Based on the above assumptions, \cite{krafty16} establishes the following theorem.

            \begin{theorem} \label{theorem1} \textbf{($\Omega_w$-norm consistent)}
                For every infinite vectors $\pmb{p_1}, \pmb{p_2},\cdots$, there exist a series of a finite eigenvectors $\pmb{\hat{p}_1},\pmb{\hat{p}_2},\cdots,\pmb{\hat{p}_q}$ such that 

            \begin{equation}\tag{3.13}\label{eq:eq3.13}
                \lvert \lvert (\pmb{\hat{p}_{q1}} \quad \pmb{\hat{p}_{q2}} \quad  \cdots \pmb{\hat{p}_q} \quad 0 \cdots)^T - \left(\pmb{p_{q1}} \quad \pmb{p_{q2}} \quad \cdots \right)^T \rvert \rvert_{\Omega_w} \xrightarrow{p} 0 
            \end{equation}
            \end{theorem}

            \noindent where $\lvert \lvert \cdot \rvert \rvert_\mathcal{K}$ is the reproduced Hilbert space with kernel $\mathcal{K}$, where for a positive definite matrix $\mathcal{K}$. This notation follows the theorem 2.1 of \cite{shin08}.

            Based on Theorem 1, we have

            \begin{corollary} As $n$, $N$, $R_N$,  $L\to\infty$, $\frac{R_N}{N} \rightarrow 0$,
                \begin{enumerate} [label=\roman*.]
                    \item $\hat{d}^{F}_{jkq} \xrightarrow{p}d_{jkq}.$
                    \item $\widehat{F}(\hat{\pmb{c}}_{*}) \xrightarrow{p}F(\pmb{c}_{*}).$
                \end{enumerate}
            \end{corollary}

            \subsection{Cepstral Estimator based on M-Periodogram}
                Most non-parametric estimators of the spectral density are based on the periodogram, which is not resistant against outliers or heavy-tailed distributions \textcolor{black}{as explained for instance in \cite{fox1972} and \cite{li2008}}. This issue has been recently discussed by many authors who have introduced new \textcolor{black}{types of periodograms} that display superiority over the classical one when the data has some atypical observations, for example, \cite{Reisen:levy:taqqu:2017,reisen2018,li2008laplace}, among others. Even the averaged periodogram described in Equation \ref{eq:eq3.11}, with small bias compared to the classical periodogram, is still affected by additive outliers, which are the most dangerous type of outlier in time series, drastically affecting the statistical quantities of the sample mean, variance, and periodogram \cite{reisen2018}, and consequently, the estimators derived directly from these quantities, such as Equation \ref{eq:eq3.11}.
            
                Given this undesirable sampling property of these estimators, we propose here a multitaper based on the $M-$periodogram estimator that can cope well with the impact of additive outliers or heavy-tailed distributions on the discriminant analysis tools.

                In this context, we consider here the $M-$periodogram by \cite{katkovnik98} and extended by \cite{reisen2018,li2008} to suggest the discriminant cepstrum  tool based on this robust periodogram, denoted here by $M$-cepstral estimator.
    
                Following the notation in \cite{reisen2018} $I_{jk}^{M} (\lambda_{m})$ is defined as
                
                \begin{equation}\tag{3.14}\label{eq:eq3.14}
                    I^{M}_{jk}(\lambda_{m}) = \frac{N}{8\pi}\left\lVert \pmb{\hat{\beta}_{jk}^{M} (\lambda_{m})} \right\rVert^{2} = \frac{N}{8\pi}\left (\hat{\beta}_{jk}^{(c)} (\lambda_{m})^{2} + \hat{\beta}_{jk}^{ (s)} (\lambda_{m})^{2}\right).
                \end{equation}

            \noindent $\pmb{\hat{\beta}_{jk}^{M}(\lambda_{m})} =  [\beta_{jk}^{(c)}(\lambda_{m}),\beta_{jk}^{(s)}(\lambda_{m})]^{T}$ is the solution of 
                
                \begin{equation}\tag{3.15}\label{eq:eq3.15}
                    \pmb{\beta^{M}_{jk}(\lambda_{m})} = \underset{\beta \in \mathbb{R}^2}{\operatorname{argmin}} \sum_{t=1}^N \psi\left[X_{jkt}-\pmb{S^{T}_{jk}(\lambda_{m})}\beta_{jk}(\lambda_{m}).\right]
                \end{equation}

                \noindent where $\pmb{S_{jk}(\lambda_{m})} = \left [\cos(t\lambda_{m}), \sin(t\lambda_{m}) \right]^T$ and $\psi(\cdot)$ is the Huber \citep{huber,mukherjee08} influence function defined as
        
                \begin{equation}\tag{3.16}\label{eq:eq3.16}
                    \psi_H (x) = \begin{cases} x, & \lvert x \rvert \leq c \\
                    c\,\text{sign}(x), & \text{ otherwise.}
                    \end{cases}
                \end{equation}

                \begin{assumptions} \label{A.3}
                    The sequence $\{\epsilon_{jkt}\}$ is now assumed to be a Gaussian White Noise process, ie.  $\{\epsilon_{jkt}\} \sim WN(0,\sigma^2)$.
                \end{assumptions}

                \begin{remark} \label{remark 5}
                  Under  Assumption 3, \cite{celine22} showed that:
                  for each $\tau$, $\hat{\gamma}^M_N(\tau) \xrightarrow{p} \frac{\gamma(\tau)}{a^2}$, where $a = 2\Phi(c)-1$, $\Phi$ is the cumulative distribution function of a standard Gaussian random variable and c is from Equation \ref{eq:eq3.16}. In addition, the authors show that $\gamma_{\psi}(\tau) \leq c \left \lvert \frac{\gamma(\tau)}{\gamma(0)} \right \rvert$, where $\gamma_{\psi}(\tau)$ and $\hat{\gamma}^M_N(\tau)$ are the theoretical ACF and its estimator derived from the use of Huber loss function, respectively. Also, $\gamma_(\tau)$ is the covariance of the process.
                \end{remark}

                \begin{remark} \label{remark 6}
                    \cite{celine22} defined the spectral density of the process using $\gamma_{\psi}(\tau)$ as

                \begin{equation}\tag{3.17}\label{eq:eq3.17}
                    S_{\psi,jk}(\lambda) = \frac{1}{2\pi} \sum_{\tau =-\infty}^{\infty} \gamma_{\psi,jk}(\tau) cos(\lambda \tau), \quad \text{for all }\lambda\in [-\pi,\pi].
                \end{equation}
                \end{remark}

                As previously, let $R=R_N$ (see Remark \ref{remark 4}), thus, multitaper M-periodogram ($I^{RM}_{jk}(\lambda_{m})$) is now  defined as
    
                \begin{equation}\tag{3.18}\label{eq:eq3.18}
                    I^{RM}_{jk}(\lambda_{m}) = \frac{1}{R_N} \sum_{r=1}^R I_{jk}^{r M} (\lambda_{m})
                \end{equation}
    
                \noindent where 
                
                \begin{equation}\tag{3.19}\label{eq:eq3.19}
                    I^{rM}_{jk}(\lambda_{m}) = \frac{N}{8\pi} \left\lVert \pmb{\hat{\beta}_{jk}^{r \hspace{.1cm} M} (\lambda_{m})} \right\rVert^{2} = \frac{N}{8\pi} \left (\hat{\beta}_{jk}^{r (c) \hspace{.1cm}} (\lambda_{m})^{2} + \hat{\beta}_{jk}^{ r (s) \hspace{.1cm}} (\lambda_{m})^{2}\right)
                \end{equation}
    
                \noindent and $\pmb{\beta_{jk}^{r \hspace{.1cm} M}(\lambda_{m})} =  [\beta_{jk}^{r(c)}(\lambda_{m}),\beta_{jk}^{r(s)}(\lambda_{m})]^{T}$ is the solution of 
                
                \begin{equation}\tag{3.20}\label{eq:eq3.20}
                    \pmb{\beta^{rM}_{jk}(\lambda_{m})} = \underset{\beta \in \mathbb{R}^2}{\operatorname{argmin}} \sum_{t=1}^N \psi\left[h_{rt}X_{jkt}-\pmb{S^{T}_{jk}(\lambda_{m})}\pmb{\beta_{jk}(\lambda_{m})}\right]
                \end{equation}
                
                To derive asymptotic properties for short and long-memory times series,  the authors \cite{reisen2018},  among others, consider   the Huber loss-function with the constant $c_{H}=1.345$. An additional Assumptions made in this paper as follows
                
                As a result, the M-discriminant function can be estimated based on $M-$cepstrum coefficients ($\hat{c}^{M}_{jk\ell}$), that can be defined as follows:
    
                \begin{equation}\tag{3.21}\label{eq:eq3.21}
                    \hat{d}^{M}_{jks} = \left\langle\,\pmb{\hat{p}}^{M}_{\ell},\pmb{\hat{c}^{(k)}_{jk\ell}}^M \right\rangle =  \sum_{\ell=0}^{L-1} \hat{p}^{M}_{s\ell}\hat{c}^{M}_{jk\ell},
                \end{equation}
    
                \noindent where $\hat{p}^{M}_{s\ell}$ are truncated estimator of the true $p_{s\ell}$.
    
                \begin{equation}\tag{3.22}\label{eq:eq3.22}
                    \hat{F}^{M}(\hat{d}^{M}_{*1},\cdots,\hat{d}^{M}_{*q};\hat{\mu}^{M}_j) = \min_{j}\left\{ \sum_{s=1}^{q}\left[ \left ( \hat{\pmb{d}}^{M}_{*s} - \hat{\pmb{\mu_{js}}^{M}} \right)\right]^{2} -2\ln (\hat{f}_{j}) \right\},
                \end{equation}
    
                \noindent where $\hat{\pmb{d}}_{jks}^{M}$ and $\hat{F}^{M}$ are the estimators of \ref{eq:eq2.20} and \ref{eq:eq2.21}, respectively. 
        
        \begin{proposition}\label{proposition 4}
            Under Assumptions\ref{A.1} and \ref{A.2}, $\hat{\gamma}^M_N(\tau) \overset{\text{p}}{\longrightarrow} \hat{\gamma}_N(\tau)$, as $N \rightarrow \infty$. 
        \end{proposition}

        \begin{proposition}\label{proposition 5}
                In \ref{eq:eq3.1}, let $\hat{S}_{\psi,jk}(\lambda_m)=$ $I^{MR_N}_{jk}(\lambda_m)$ (\ref{eq:eq3.11}), where $R=R_N$ is defined in Remark \ref{remark 4}. Thus,  $\hat{c}^M_{jk\ell} \xrightarrow{p} c_{jk\ell}$ and $\hat{d}^M_{jk\ell} \xrightarrow{p} d_{jk\ell}$ when $\frac{R_N}{N} \rightarrow 0$ and $N,R_N \rightarrow \infty$.
            \end{proposition}

        \begin{proof}
            Using the results of Proposition 4, the stationarity of the process allows us to use the cosine function for the Fourier transform, and the logarithm is a continuous function. Then, by applying the mapping theorem (Theorem 9.4.2 from \cite[p. 541]{anderson76}), the results follow directly.
        \end{proof}

        


        \begin{proposition} As $n$, $N$, $R_N$,  $L\to\infty$, $\frac{R_N}{N} \rightarrow 0$, \label{proposition 6}
            $\left\lVert  \hat{d}^M_{jks}-\hat{d}_{jks} \right\rVert = o_p(1)$
        \end{proposition}

        \begin{theorem} \label{theorem2}
            Under the Assumption\ref{A.1} and using sine tapers as defined in Equation \ref{eq:eq2.3}, each data taper are orthogonal and the periodogram are independent for each $r$ taper. Then, $I^{RM}_{jk}(\lambda_{m})\overset{\text{p}}{\longrightarrow} S_x(\lambda_m)$ for all $\lambda_m$, as $R_N\rightarrow \infty$.
        \end{theorem}

        \begin{lemma} \label{lemma2}
            $cov\left(I_{jk}^{MR}(\lambda_m),I_{jk}^{MR}(\lambda_k) \right) \rightarrow 0$, for all $k\neq m$ and $R_N \rightarrow \infty$. 
        \end{lemma}


    \section{Simulations}

    \subsection{$AR(1)$ simulation with fixed parameters}

        Before show the results of Monte Carlo simulation, an example is conducted of the empirical cepstral function for the AR$(1)$ model $X_{jkt} = \phi_{jk}X_{jk(t-1)} + \epsilon_{jkt}$, where $\epsilon_{jkt} \sim \mathcal{N}(0,\sigma^2_{\epsilon \hspace{.1cm}jk})$.
        We considered in the simulation  $J=3$, $n_j=1$, $N = 100$, the AR parameters $\phi \in \{0.25,0.50,0.75\}$ and $\sigma^2_{\epsilon \hspace{.1cm} jk} = 1$.
        The classical periodogram was used in this example as the estimator of the spectral density and the results are presented in Figure \ref{fig:fig5}. There are three panels for the time domain: (a) time series replicates; (b) ACF; and (c) PACF; and three panels for the frequency domain: (d) Periodogram; (e) ln-Periodogram; and (f) estimated Cepstra. 

        \begin{figure}[H]
            \centering
            \includegraphics[width=1\linewidth, height=10cm]{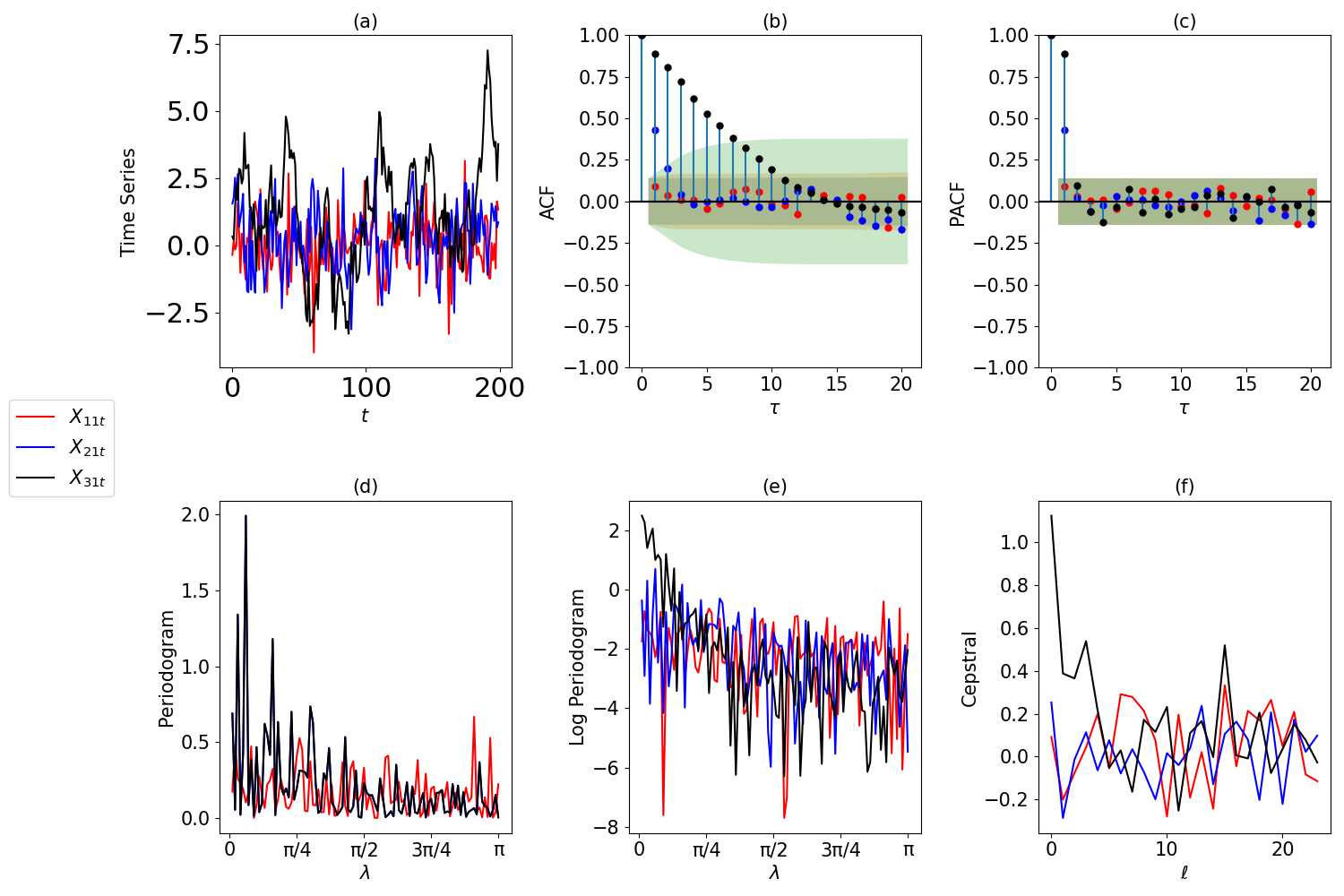}
            \caption{ $\phi \in [0.1,0.5,0.9]$ and $\sigma^2_{\epsilon\hspace{.1cm}jk} = 1$: (a) Time Series, (b) ACF, (c) PACF, (d) Periodogram, (e) ln-Periodogram, and (f) Cepstra}
            \label{fig:fig5}   
        \end{figure}

        \begin{figure}[H]
            \centering
            \includegraphics[width=1\linewidth, height=7cm]{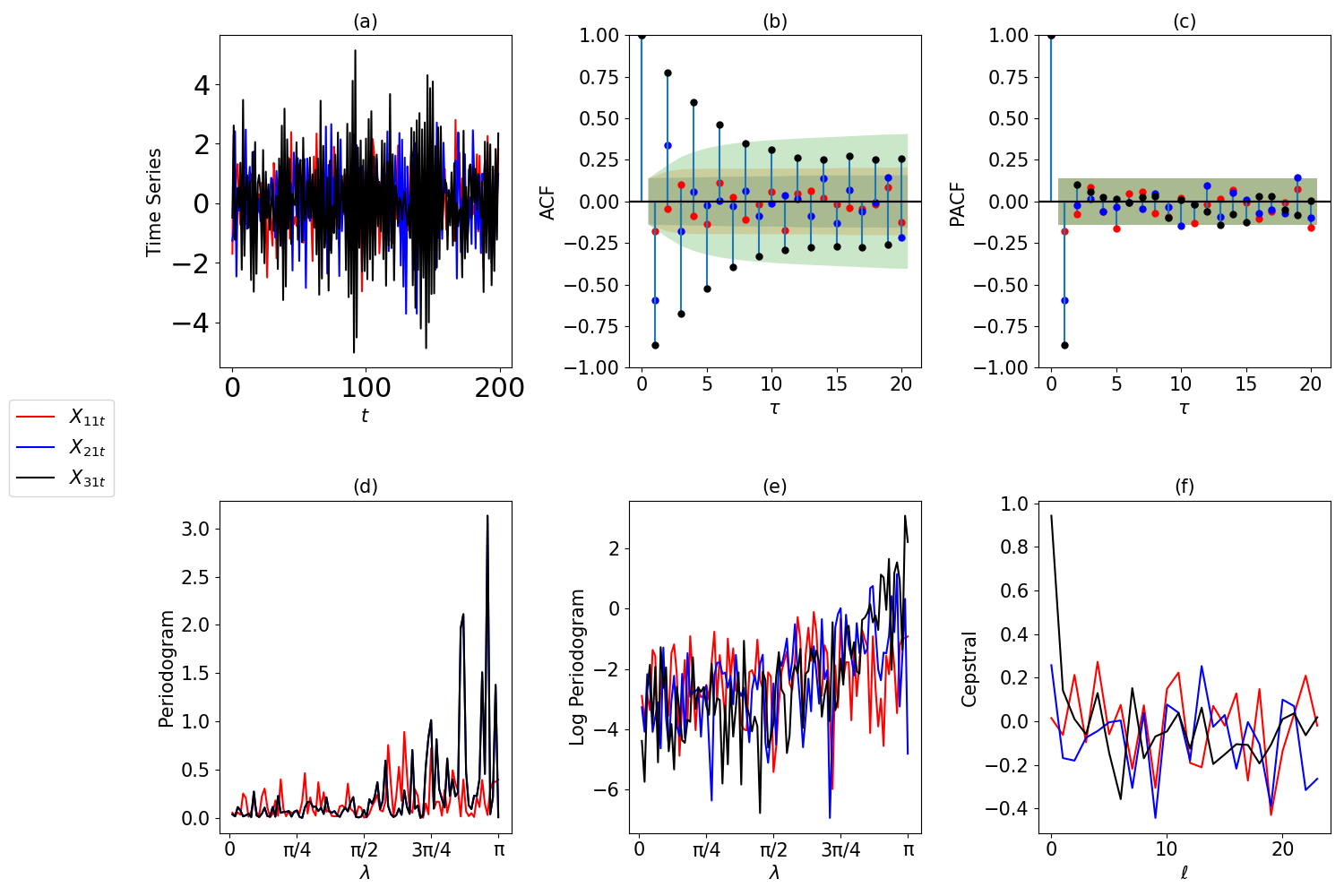}
            \caption{ $\phi \in [-0.1,-0.5,-0.9]$ and $\sigma^2_{\epsilon\hspace{.1cm}jk} = 1$: (a) Time Series, (b) ACF, (c) PACF, (d) Periodogram, (e) ln-Periodogram, and (f) Cepstra}
            \label{fig:fig6}   
        \end{figure}

        \begin{figure}[H]
            \centering
            \includegraphics[width=1\linewidth, height=10cm]{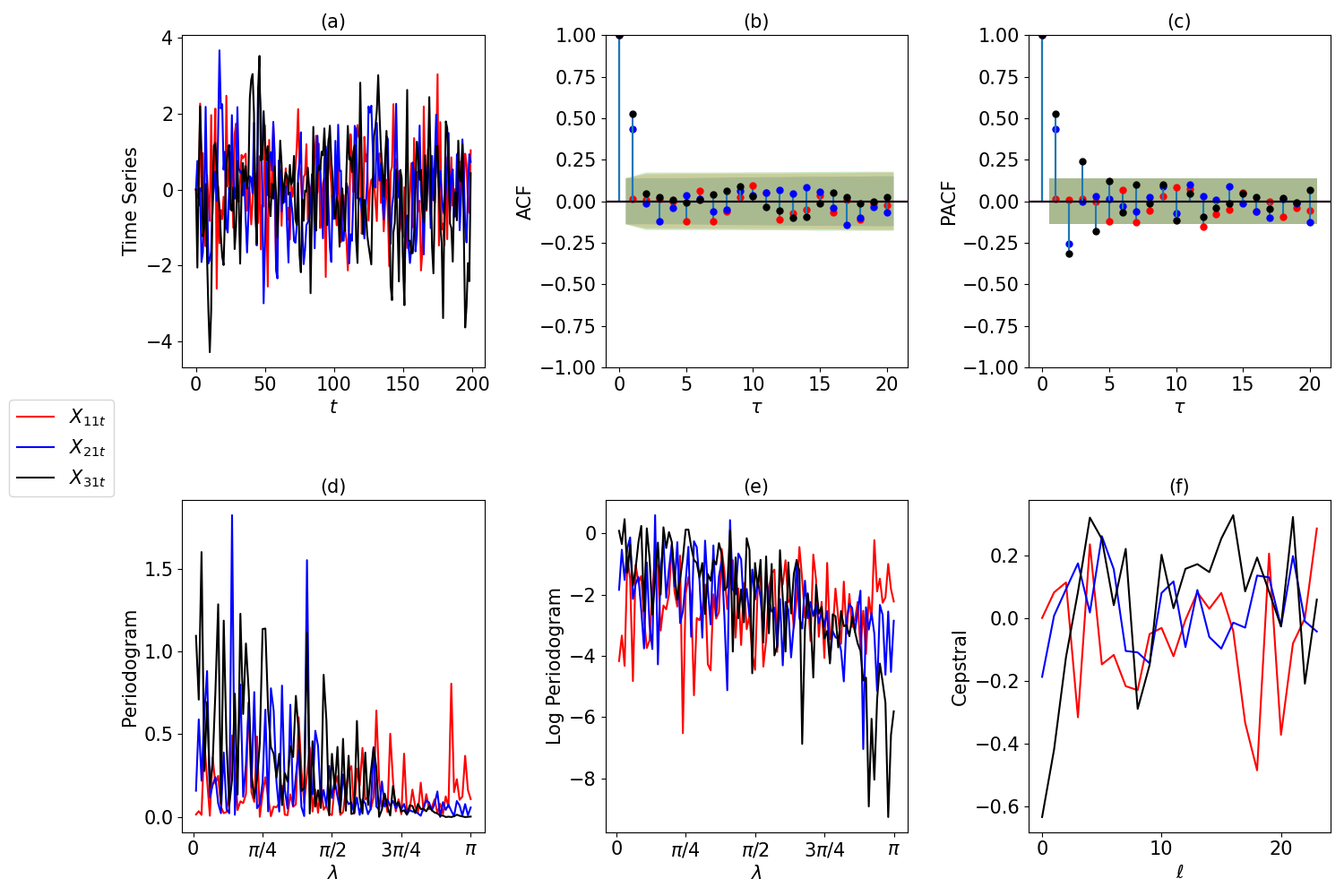}
            \caption{ $\theta \in [0.1,0.5,0.9]$ and $\sigma^2_{\epsilon\hspace{.1cm}jk} = 1$: (a) Time Series, (b) ACF, (c) PACF, (d) Periodogram, (e) ln-Periodogram, and (f) Cepstra}
            \label{fig:fig7}   
        \end{figure}

        \begin{figure}[H]
            \centering
            \includegraphics[width=1\linewidth, height=10cm]{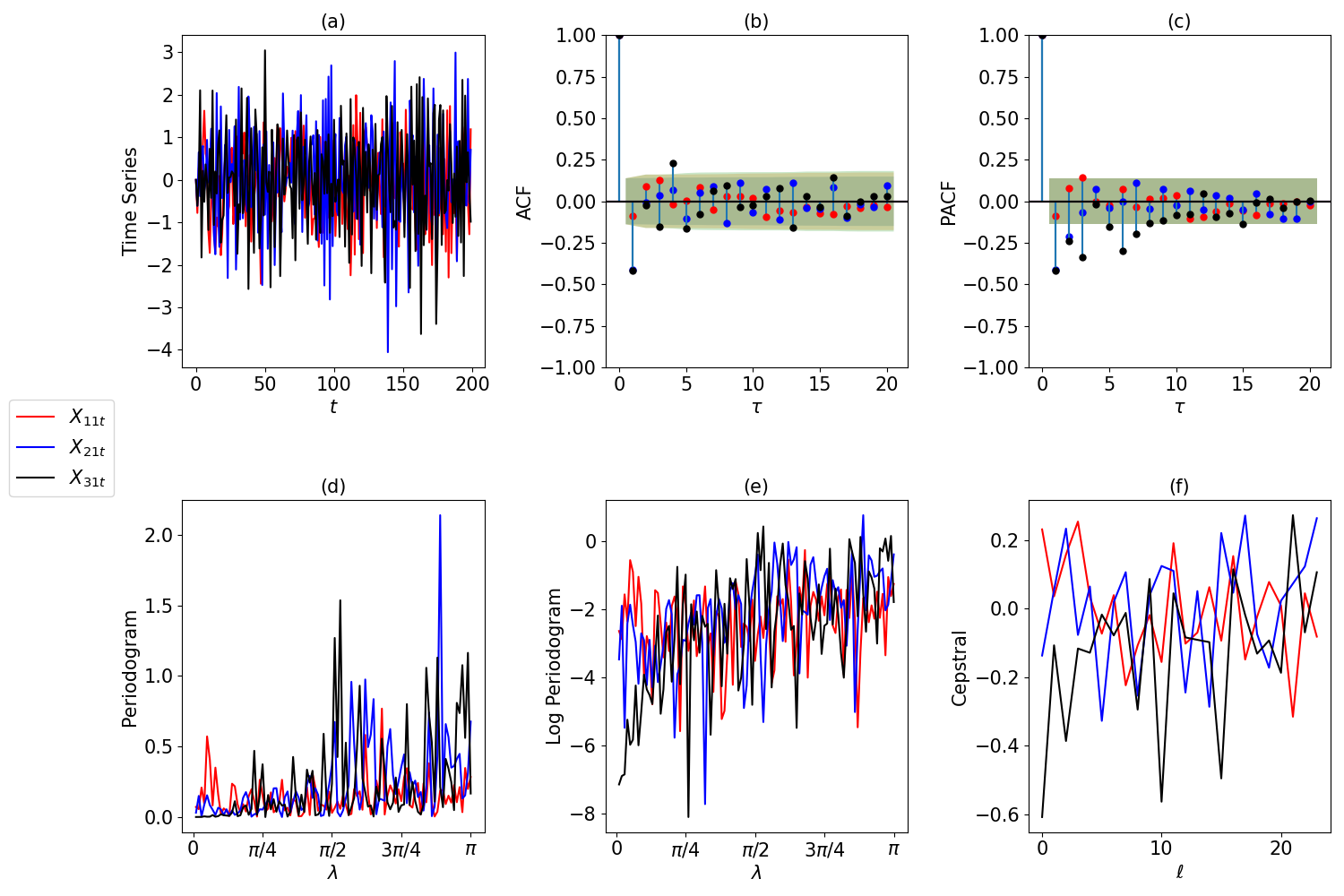}
            \caption{ $\theta \in [-0.1,-0.5,-0.9]$ and $\sigma^2_{\epsilon\hspace{.1cm}jk} = 1$: (a) Time Series, (b) ACF, (c) PACF, (d) Periodogram, (e) ln-Periodogram, and (f) Cepstra}
            \label{fig:fig8}   
        \end{figure}



        On the one hand, the time domain is shown in panels (a), (b), and (c). It is possible to note that increasing the parameter will increase the variability for constant variance noise. This also influences the autocorrelations, which become more persistent over time as the coefficient increases. All coefficients can be seen in the partial ACF with just one significant peak.
                
        On the other hand, the frequency domain is represented in panels (d) and (e). It can be noted that the estimated spectra of an AR$(1)$ process exhibit lower frequency components. This means that the most variability of the process is due to the frequencies next to zero. Another important feature is the fact that for frequencies greater than 1, the time series cannot be differentiated. According to \cite{krafty16}, there is a contribution of this kind of frequencies that have amplitudes next to zero.
                
        To observe more variability and to discriminate better, the logarithm is calculated on the periodogram, generating the graphs presented in panel (e). Clearly, they are now more easily distinguishable. More complicated processes are even better taken into account with this variability.

    \subsection{Monte Carlo Simulation using AR$(2)$}

        A Monte Carlo simulation was performed to analyze the sample properties of the M-cepstra and its impact on classification rates in the context of within spectral variability with and without contamination of abrupt observations that behave as additive outliers. For the simulations, $J=3$ populations are considered in 6 different situations: $n_j = \{15,50,100\}$ and $N=\{1000,2000\}$. In each scenario, the classification rate is calculated and the scenarios are repeated 100 times to obtain the mean and standard deviation of the classifications. In order to compare with \cite{krafty16} and due to the characteristics of the data application presented in Section 5, the simulations are conducted using a conditional (on population) AR$(2)$ model that can be represented by
        
        \begin{equation}\tag{4.1}\label{eq:eq4.1}
            \begin{matrix}
                X_{jkt} =\phi_{jk1}X_{jkt-1} + \phi_{jk2}X_{jkt-2} + \epsilon_{jkt} & \epsilon_{jkt} \sim N(0,\sigma_{jk}^2)
            \end{matrix}
        \end{equation}
        
        \noindent where the parameter for each population is a random variable presented in Table \ref{tab:tab1}.

        \begin{table}[H]
            \centering
            \caption{Conditional Parameters Distribution}
            \label{tab:tab1}
            \begin{adjustbox}{width=\textwidth}
                \begin{tabular}{c|c|c|c|c} 
                \hline
                \diagbox{\textbf{Population}}{\textbf{Parameters}} & $\phi_{jk1}$                        & $\phi_{jk2}$                       & \textbf{White Noise}                    & \textbf{White Noise's Variance}    \\ 
                \hline
                $j=1$                                              & $(\phi_{1k1} |j=1) \sim UNI(0.05,0.7)$ & $(\phi_{1k2}  |j=1) \sim UNI(-0.12,-0.06)$ & $\epsilon_{1k} \sim N(0,\sigma_{1k}^2)$ & $ (\sigma_{1k}^2  |j=1) \sim UNI(0.1,10)$   \\
                $j=2$                                              & $(\phi_{2k1} |j=2) \sim UNI(0.01,1.2)$ & $(\phi_{2k2}  |j=2) \sim UNI(-0.36,-0.25)$ & $\epsilon_{2k} \sim N(0,\sigma_{2k}^2)$ & $(\sigma_{2k}^2 |j=2)  \sim UNI(0.3,3.0)$    \\
                $j=3$                                              & $(\phi_{3k1} |j=3) \sim UNI(0.12,1.5)$ & $(\phi_{3k2} |j=3) \sim UNI(-0.75,-0.56)$ & $\epsilon_{3k} \sim N(0,\sigma_{3k}^2)$ & $(\sigma_{3k}^2 |j=3) \sim UNI(0.9,1.1)$  \\
                \hline\hline
                \end{tabular}
            \end{adjustbox}
        \end{table}
        
        The conditional AR$(2)$ model is chosen with parameters being random variables with uniform distribution in overlapping intervals. Additionally, the disturbances are modeled as normal distribution with zero mean and their variance as a uniform distribution with different increasing intervals. This characteristic is important to create time series that can be mixed and create difficulty for discrimination and classification, as would occur in a real application context.
        
        \sloppy
        Each scenario includes all three populations for each time series and replicas. All transformations are performed, including multitaper periodogram estimation for both classical and $M-$periodogram. The cepstral vector is then obtained for each replicate within its own population. An LDA parameter is estimated based on the cepstra and used for the classification test. For this test, a new group of time series for each population is simulated with the same parameters. For comparison with \cite{krafty16}, $N=50$ is chosen for the testing data. These new time series undergo the same transformations as the original replicas until achieving the cepstral vector, which is then used for testing.
        
        Finally, the simulations are conducted in two situations: without outliers, where the model is fitted without extreme observations, and with outliers, where the replicas are contaminated with an additive outlier following \cite{huber,reisen20} with
        
        \begin{equation}\tag{4.2}\label{eq:eq4.2}
        Z_{jkt}=X_{jkt}+ \omega I_{jkt},
        \end{equation}
        
        \noindent where $\omega$ is the magnitude of the outlier, $\{X_{jkt}\}$ and $\{I_{jkt}\}$ are independent, and $\{I_{jkt}\}$ is a sequence of IID random variables with $P\left(I_{jkt}=-1\right)=P\left(I_{jkt}=1\right)=p/2$ and $P\left(I_{jkt}=0\right)=1-p$, $p \in (0,1)$. Here, we take $p=0.01$ and $\omega=7$. The uncontaminated case corresponds to $\omega=0$. The proposed robust cepstral discriminant analysis was compared to the cepstral multitaper discriminant analysis, both utilizing $R=7$ tapers. The tuning constant in the Huber model is set to $1.345$ as the asymptotically optimal value defined by \cite{reisen-2020-cyclost}. A test data set of 50 time series per group is simulated for each random sample to evaluate the classification rate of the training data. The correct classification rate is calculated in each scenario using the LDA model with cepstral coefficients. This scenario was repeated 100 times and the mean and standard deviation of the correct classification rates are calculated and presented in Table \ref{tab:tab2}.
        
        The performance of the two classification methods - Classical and Huber - are evaluated in different scenarios, with and without the presence of outliers. It is noted that for samples without outliers, the results show that both classification methods exhibit high correct classification rates. In $\Pi_1$, for instance, with $n_j = 15$, the average correct classification rates were 89.64\% and 90.09\% for the Classical method with $N = \{1000,2000\}$, respectively. For the Huber method, the rates were slightly lower, with 85.37\% and 89.85\% for $N =\{1000,2000\}$, respectively.
        
        As the number of replicas per population increases, an improvement in correct classification rates is observed for both methods. With $n_j$ = 100, the Classical and Huber methods achieved 91.02\% and 93.00\% correct classification for $N = 1000$, respectively. The same trend is observed for $N = 2000$, where the rates reach 93.00\% and 92.19\%.
        
        In the presence of outliers, a decrease in correct classification rates is noted for both methods. For $\Pi_1$ with $n_j$ = 15, the average correct classification rates dropped to 67.06\% and 68.52\% for the Classical method with N = 1000 and 2000, respectively. However, the Huber method proved to be   more robust, with rates of 79.13\% and 83.64\% for N = 1000 and 2000, respectively.
        
        The robustness of the Huber method is more evident as the number of replicas per population increases. With $n_j$ = 100, the correct classification rates for the Classical method were 76.18\% and 82.78\%, while the Huber method achieved 81.06\% and 88.03\% for N = 1000 and 2000, respectively. This indicates that the Huber method is more effective in the presence of outliers when compared to the Classical method.
        
        Comparing the different variance scenarios ($\sigma^2$), it is observed that for $\Pi_1$ (0.01 - 10), the correct classification rates are generally lower than for $\Pi_2$ (0.3 - 3) and $\Pi_3$ (0.9 - 1.1). For example, in $\Pi_3$, for $n_j$ = 15, the correct classification rates without outliers were 93.67\% and 93.88\% for the Classical method with $N = \{1000,2000\}$, respectively, showing better performance compared to $\Pi_1$ and $\Pi_2$.
        
        With the inclusion of outliers, $\Pi_3$ also showed overall better robustness. The correct classification rates for the Huber method in $\Pi_3$ were consistently higher compared to $\Pi_1$ and $\Pi_2$. This suggests that classifier performance can be strongly influenced by the variance of the data.
        
        \begin{table}[H]
            \centering
            \caption{Monte Carlo simulation: mean and standard deviation rate for correct classification}
            \label{tab:tab2}
            \begin{adjustbox}{width=\textwidth}
                \begin{tabular}{c|c|c|cccc|cc|cccc} 
                \hline
                \diagbox{}{$\sigma_{jk}^2$}  & $n_j$ & \multicolumn{6}{c|}{WITHOUT OUTLIER} & \multicolumn{4}{c}{WITH OUTLIER} \\ 
                \hline
                $I_{jk}$  & ~- & \multicolumn{4}{c|}{Classic} & \multicolumn{2}{c|}{M} & \multicolumn{2}{c}{Classic} & \multicolumn{2}{c}{M} \\ 
                \hline
                $N$ & ~- & 250 & 500 & 1000 & 2000 & 1000 & 2000 & 1000 & 2000 & 1000 & 2000 \\ 
                \hline
                \multirow{6}{*}{$\Pi_1$} & \multirow{6}{*}{(0.01 - 10)} & \multirow{2}{*}{15} & 83.94\% & 86.97\% & 89.64\% & 90.09\% & 85.37\% & 89.85\% & 67.06\% & 68.52\% & 79.13\% & 83.64\% \\
                                          &                              &                      & (6.72\%) & (6.40\%) & (5.80\%) & (5.96\%) & (6.10\%) & (5.99\%) & (9.10\%) & (8.45\%) & (7.45\%) & (7.91\%) \\
                                          &                              & \multirow{2}{*}{50} & 85.55\% & 88.34\% & 90.85\% & 91.02\% & 86.40\% & 90.13\% & 75.08\% & 81.10\% & 80.34\% & 84.74\% \\
                                          &                              &                      & (5.58\%) & (5.17\%) & (4.47\%) & (4.63\%) & (5.20\%) & (4.53\%) & (4.97\%) & (3.89\%) & (4.89\%) & (4.48\%) \\
                                          &                              & \multirow{2}{*}{100} & 86.00\% & 88.71\% & 91.02\% & 93.00\% & 89.02\% & 92.19\% & 76.18\% & 82.78\% & 81.06\% & 88.03\% \\
                                          &                              &                      & (5.40\%) & (4.98\%) & (4.19\%) & (3.75\%) & (6.28\%) & (3.32\%) & (3.77\%) & (3.12\%) & (3.39\%) & (3.18\%) \\ 
                \hline
                \multicolumn{1}{c}{~} & \multicolumn{1}{c}{~} & \multicolumn{1}{c}{~} & ~ & ~ & ~ & \multicolumn{1}{c}{~} & ~ & \multicolumn{1}{c}{~} & ~ & ~ & ~ & ~ \\ 
                \hline
                \multirow{6}{*}{$\Pi_2$} & \multirow{6}{*}{(0.3 - 3)} & \multirow{2}{*}{15} & 79.25\% & 84.79\% & 88.80\% & 91.40\% & 83.10\% & 90.01\% & 72.28\% & 83.02\% & 85.26\% & 87.83\% \\
                                          &                              &                      & (8.09\%) & (6.89\%) & (6.07\%) & (5.27\%) & (4.20\%) & (4.63\%) & (6.44\%) & (5.70\%) & (5.01\%) & (4.73\%) \\
                                          &                              & \multirow{2}{*}{50} & 82.92\% & 88.02\% & 91.98\% & 92.33\% & 89.60\% & 92.02\% & 84.78\% & 86.66\% & 87.12\% & 88.43\% \\
                                          &                              &                      & (6.01\%) & (5.16\%) & (4.17\%) & (4.17\%) & (4.60\%) & (5.12\%) & (4.77\%) & (5.00\%) & (4.23\%) & (4.22\%) \\
                                          &                              & \multirow{2}{*}{100} & 83.76\% & 88.69\% & 92.53\% & 95.27\% & 91.14\% & 95.03\% & 71.04\% & 87.44\% & 89.02\% & 90.95\% \\
                                          &                              &                      & (5.44\%) & (4.76\%) & (3.91\%) & (3.21\%) & (1.74\%) & (2.59\%) & (2.97\%) & (2.84\%) & (1.93\%) & (2.71\%) \\ 
                \hline
                \multicolumn{1}{c}{~} & \multicolumn{1}{c}{~} & \multicolumn{1}{c}{~} & ~ & ~ & ~ & \multicolumn{1}{c}{~} & ~ & \multicolumn{1}{c}{~} & ~ & ~ & ~ & ~ \\ 
                \hline
                \multirow{6}{*}{$\Pi_3$} & \multirow{6}{*}{(0.9 - 1.1)} & \multirow{2}{*}{15} & 87.73\% & 91.01\% & 93.67\% & 93.88\% & 91.15\% & 93.45\% & 81.94\% & 81.08\% & 89.43\% & 91.14\% \\
                                          &                              &                      & (6.90\%) & (6.47\%) & (5.22\%) & (6.15\%) & (5.40\%) & (4.33\%) & (6.61\%) & (6.63\%) & (5.81\%) & (4.91\%) \\
                                          &                              & \multirow{2}{*}{50} & 90.38\% & 93.11\% & 95.22\% & 95.83\% & 93.10\% & 95.16\% & 84.10\% & 88.52\% & 90.56\% & 93.26\% \\
                                          &                              &                      & (4.92\%) & (4.18\%) & (3.44\%) & (4.57\%) & (3.57\%) & (4.23\%) & (3.09\%) & (4.57\%) & (3.89\%) & (2.81\%) \\
                                          &                              & \multirow{2}{*}{100} & 90.48\% & 93.55\% & 95.30\% & 96.51\% & 94.98\% & 95.71\% & 84.22\% & 93.22\% & 91.90\% & 95.23\% \\
                                          &                              &                      & (4.33\%) & (3.61\%) & (3.14\%) & (2.74\%) & (2.86\%) & (2.93\%) & (2.44\%) & (2.91\%) & (1.93\%) & (1.95\%) \\
                \hline\hline
                \end{tabular}
            \end{adjustbox}
        \end{table}

        The results of the Monte Carlo simulations indicate that the Huber method is generally more robust in the presence of outliers compared to the Classical method. Furthermore, the performance of the classification methods is influenced by the sample size and the number of replicas per population. In high variance scenarios ($\Pi_1$), both methods exhibit lower correct classification rates, while in low variance scenarios ($\Pi_3$), the classification rates are higher. These findings highlight the importance of considering the robustness of classification methods in different contexts of variance and the presence of outliers.

    \newpage

    \section{The Real Data Application: Neurodegenerative Diseases}

        According to \cite{hausdoff00}, gait patterns can provide tools to aid in the detection and monitoring of neurodegenerative diseases. For example, Amyotrophic Lateral Sclerosis (ALS) is a disease caused by the loss of motor neurons that directly affect balance. Another disease commonly known in the literature is Huntington's disease, which degenerates parts of the brain responsible for controlling movement smoothness. As a consequence, individuals with these diseases tend to have more abrupt movements, affecting their gait pattern.
        
        All these diseases affect gait and mobility, which can be identified by the cycle of walking patterns recorded by sensors in patients' feet. \cite{hausdoff00} conducted a study for this purpose, involving 49 participants who were asked to walk. Among all the participants, 16 were healthy control individuals, 20 had Huntington's disease, and 13 had ALS \footnote{There are still 16 individuals with Parkinson's disease not considered in the sample by \cite{krafty16}. For this reason, they were not analyzed in this study either for comparison purposes.}. The recorded information includes the time spent for each individual to complete a stride cycle per second \footnote{And because the data are given in seconds, the recorded data can be expressed in Hertz, representing here a measure of the frequency of strides per second.}.
        
        The participants walked for exactly 6 minutes, resulting in a total of 360 observations in seconds. The objective is to capture how much time each person takes per second to complete a stride cycle, that is, how much time they take to place their foot on the ground and do so again during walking. Therefore, the data are unit of time per unit of time, in the sense that it evaluates how many milliseconds an individual took to complete a walking cycle (pushing one foot against the ground and doing so again) for each second. The first 20 observations (first 20 seconds) were excluded from the sample by \cite{hausdoff00} as the author considers them merely a start-up effect. The original data can be found at \textcolor{blue}{https://www.physionet.org/content/gaitndd/1.0.0/}.
        
        According to \cite{krafty16}, the original data required some treatment. Firstly, the author decided to work with 3.5 minutes, which gives 210 observations out of the total 360 minutes, due to the slower pace of sick individuals compared to healthy ones, resulting in fewer samples. As this was still insufficient, two patients each from ALS and Huntington's were eliminated due to their missing observations or abrupt ones among those who remained, leaving 11 and 18 patients, respectively. Additionally, in order to avoid further eliminations, the author opted to fill in the missing observations using cubic interpolation. These observations were also used to smooth spline sampled at 2Hz, effectively doubling the amount of information to 420 observations. Finally, the analyzed dataset contains 420 observations and 45 remaining participants \citep[p. 445]{krafty16}.
        
        Note that several transformations are needed to model could be used in frequency domain and to bring it back to the time domain with the main information in order to be used in linear discriminant analysis. In this work, the main supposition is that less transformation can be used for better classification, especially not using filter for extreme observations.
        
        The main idea of this work centers on the fact that so many transformations can significantly influence the results, masking or generating false success rates. Furthermore, the fact of eliminating sick individuals seems to be a convenience choice once the healthy ones are more sTable and there are more observations.
        
        Therefore, in this work, it is chosen not to filter the original data for abrupt observations, arguing that this information only behaves as additive outliers, but is in fact relevant information for identifying sick individuals. Additionally, all individuals are considered giving $n=49$ participants and $N=120$ observations per individual (representing $2.05$ minutes). Subsequently, for comparison reasons, two cases are considered: i) \textbf{modified} data extracting abrupt observations by filtering using 3 standard deviation median exactly as done by \cite{krafty16}; ii) \textbf{non-modified} data without extracting the extreme observations. Next, in both cases, the resulting data are detrended.
        
        In this context, $J=3$ populations (1: control; 2: ALS; 3: Huntington) are being considered, with $n_1=16$, $n_2=13$, and $n_3=20$ being their respective numbers of replicates. Figures \ref{fig:9}, \ref{fig:10}, and \ref{fig:11} present the time series, the boxplot, and ACF for 9 individuals, 3 for control healthy, 3 for ALS, and 3 for Huntington's with data \textbf{modified}. Similarly, figures \ref{fig:12}, \ref{fig:13}, and \ref{fig:14} are the data for \textbf{non-modified} data.
        
        \begin{figure}[H]
        \centering
        \includegraphics[width=\textwidth]{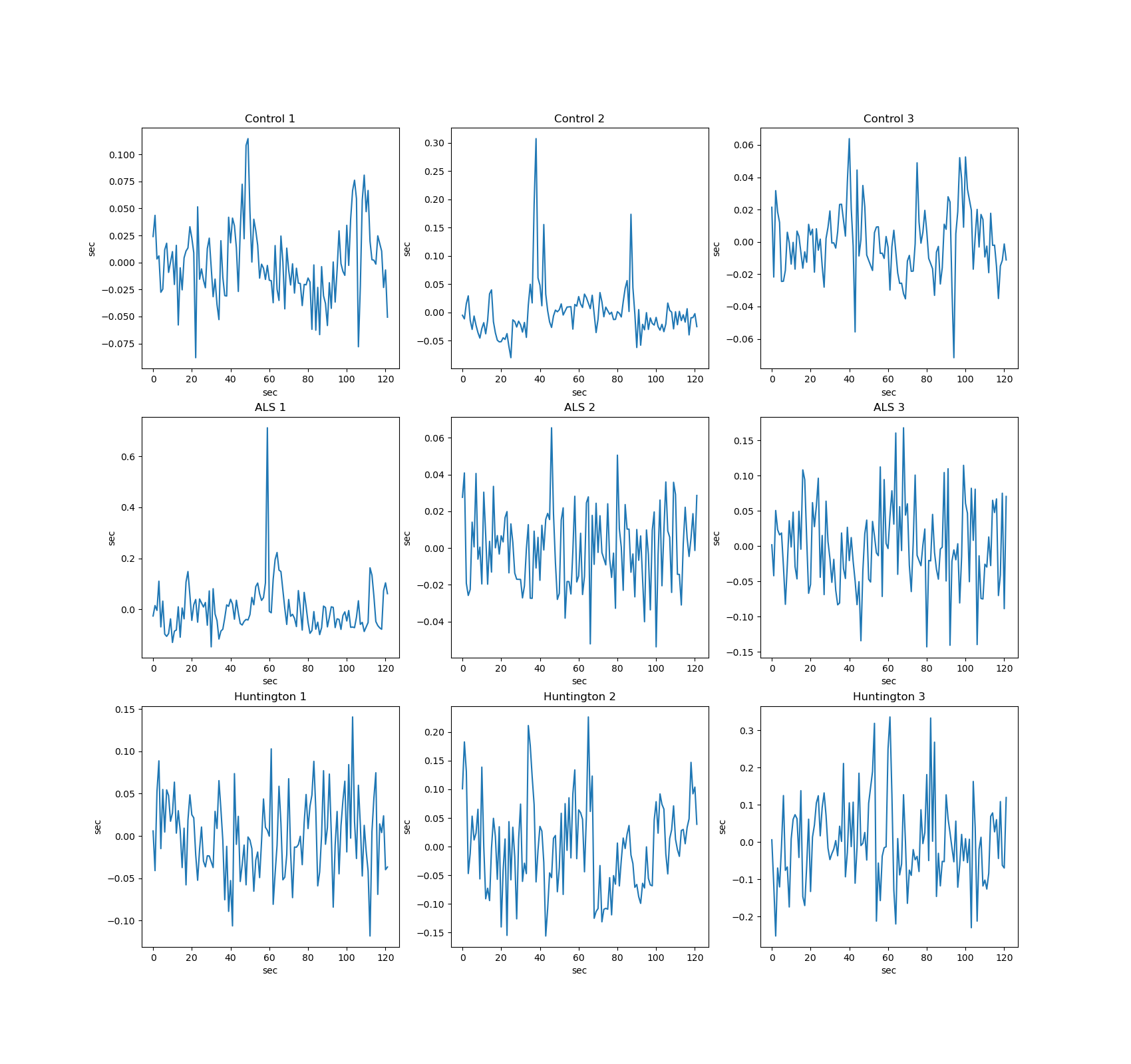}
            \caption{\textbf{Modified} detrended stride interval time series from 9 participants in neurodegenerative disease study with 3 healthy controls, 3 ALS and 3     Huntington's} \label{fig:9}
        \end{figure}
        
        \begin{figure}[H]
        \centering
            \includegraphics[width=\textwidth]{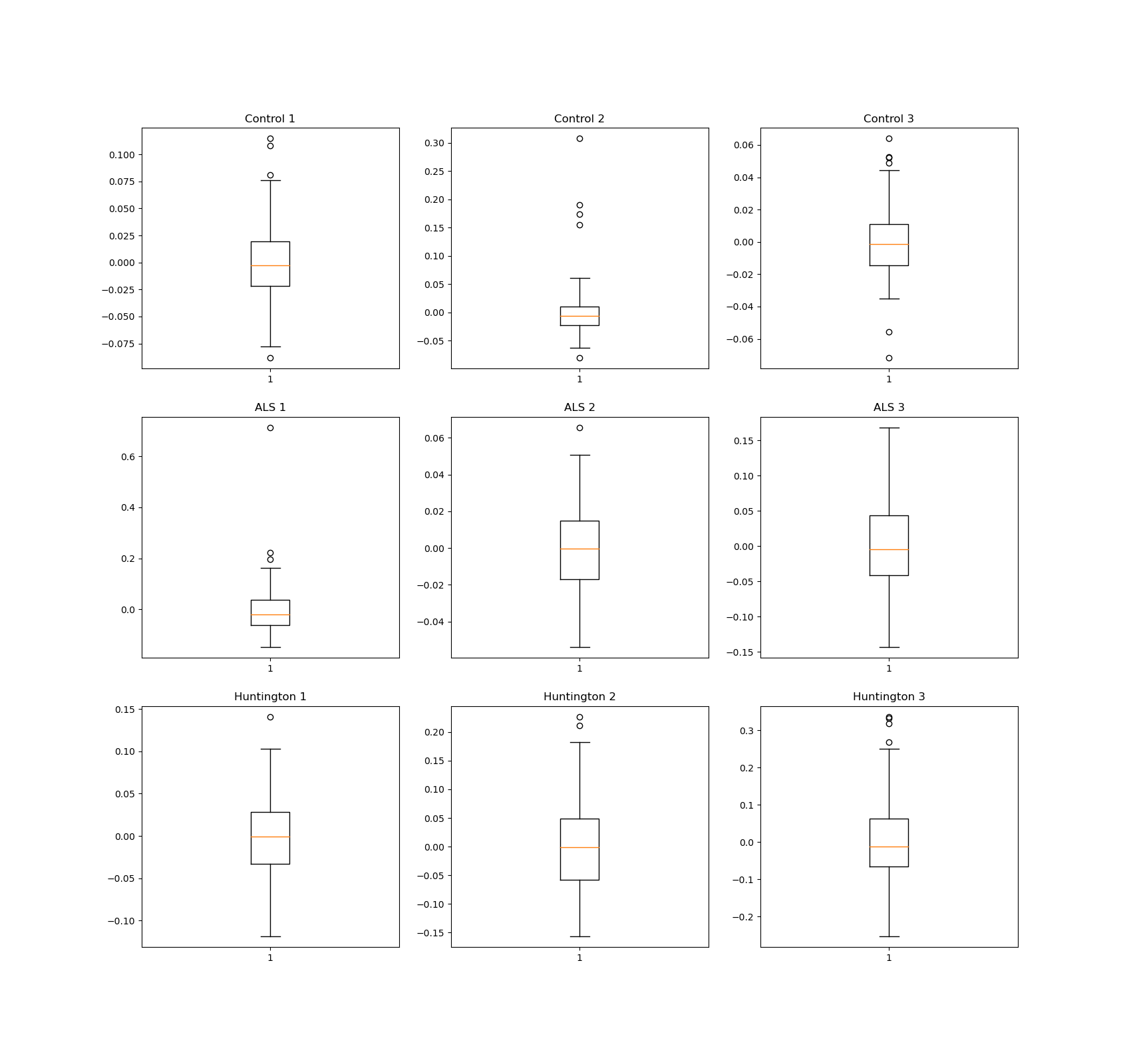}
            \caption{Box plot of \textbf{modified} detrended stride interval time series from 9 participants in neurodegenerative disease study with 3 healthy controls, 3 ALS and 3 Huntington's}
            \label{fig:10}
        \end{figure}
        
        \begin{figure}[H]
        \centering
        \includegraphics[width=\textwidth]{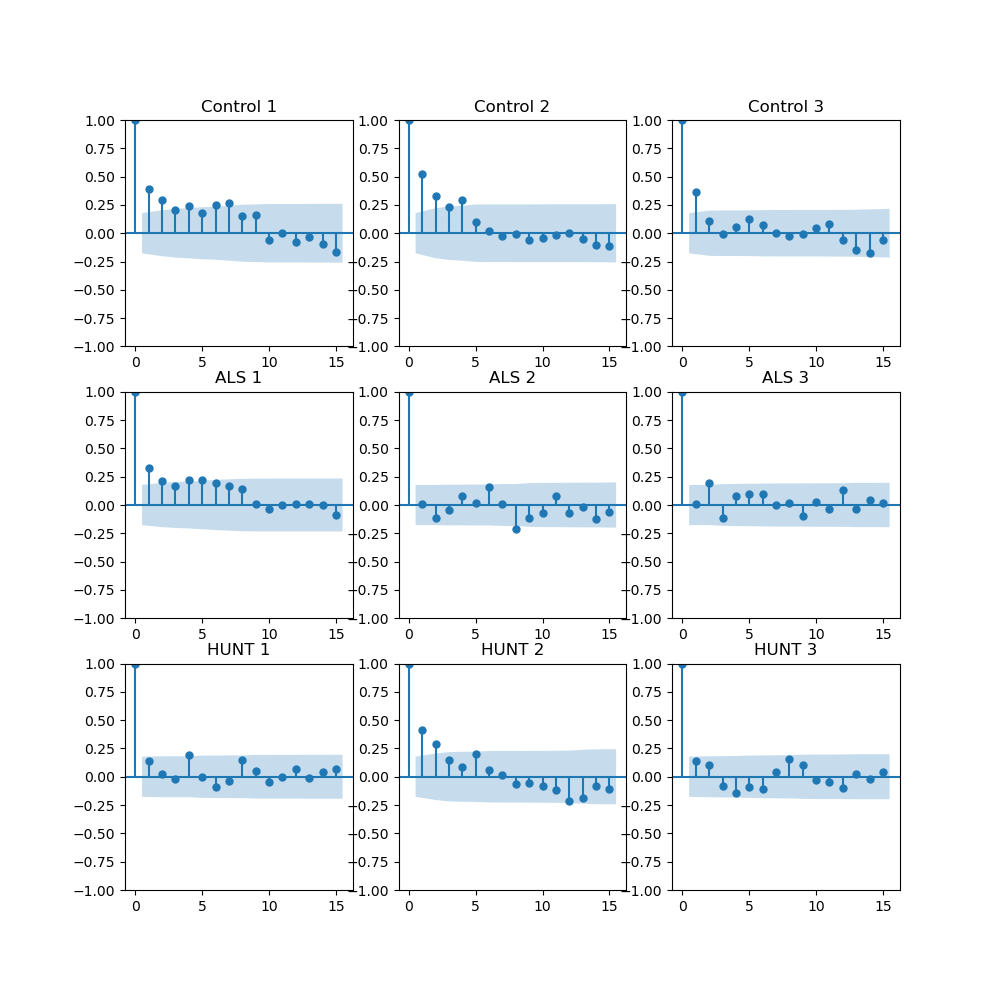}
            \caption{ACF of \textbf{modified} detrended stride interval time series from 9 participants in neurodegenerative disease study with 3 healthy controls, 3 ALS and 3 Huntington's}\label{fig:11}
        \end{figure}
        
        Comparing Figure \ref{fig:9} with \ref{fig:12}, the modified data still have abrupt observations after filtering, as can be seen in Figure \ref{fig:10}. Furthermore, healthy control individual 2, ALS 1, and ALS 3 exhibit values so extreme that they significantly affect the visualization of the data structure. Additionally, as observed in Figures \ref{fig:11}, the autocorrelation structure of the same time series differs. In other words, the mechanical removal of extreme observations affects the autocorrelation structures, which can lead to spurious conclusions when using the periodogram based on such information.
        
        \begin{figure}[H]
        \centering
        \includegraphics[width=\textwidth]{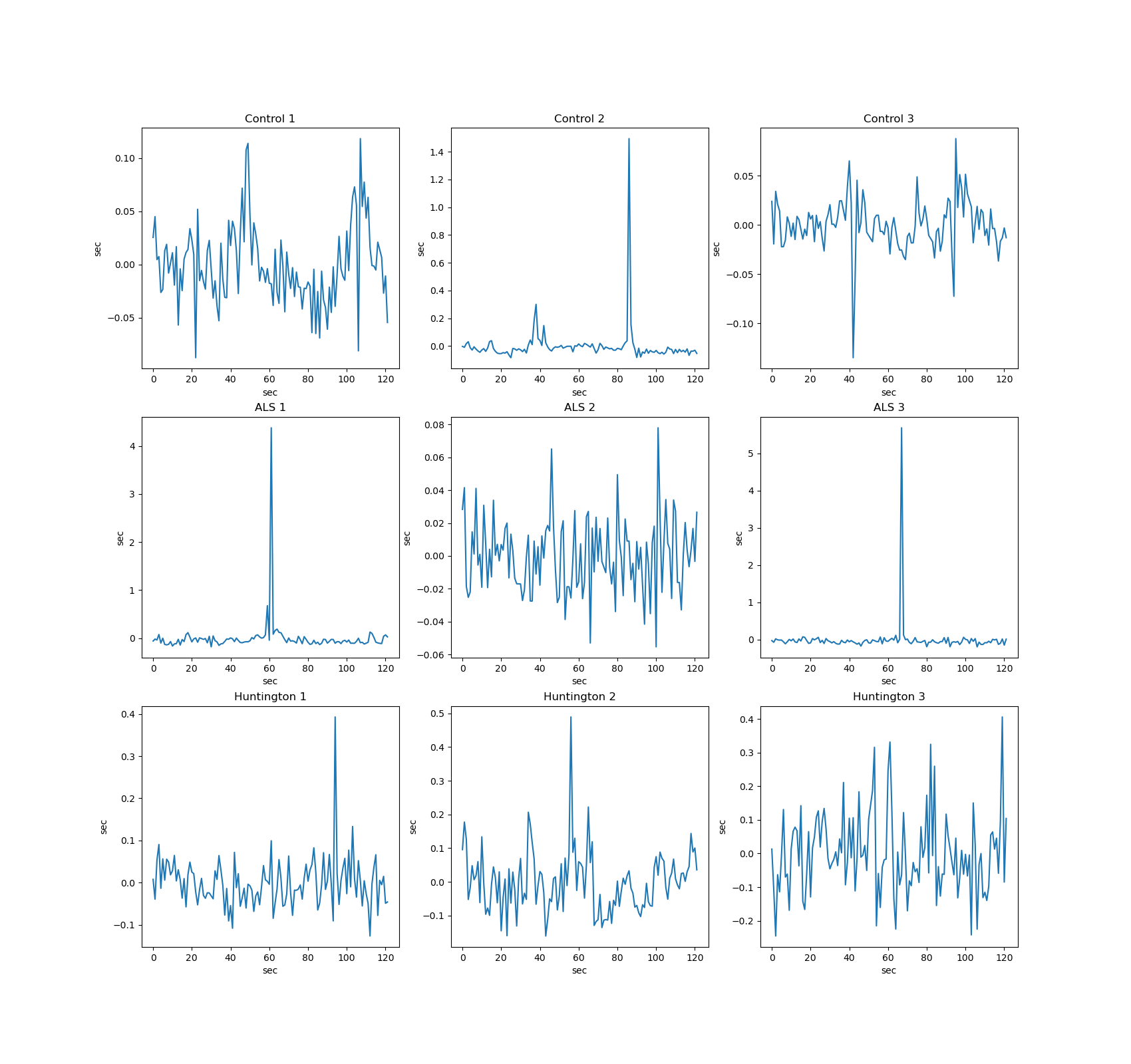}
            \caption{\textbf{Non modified} detrended stride interval time series from 9 participants in neurodegenerative disease study with 3 healthy controls, 3 ALS and 3 Huntington's}\label{fig:12}
        \end{figure}
        
        \begin{figure}[H]
        \centering
            \includegraphics[width=\textwidth]{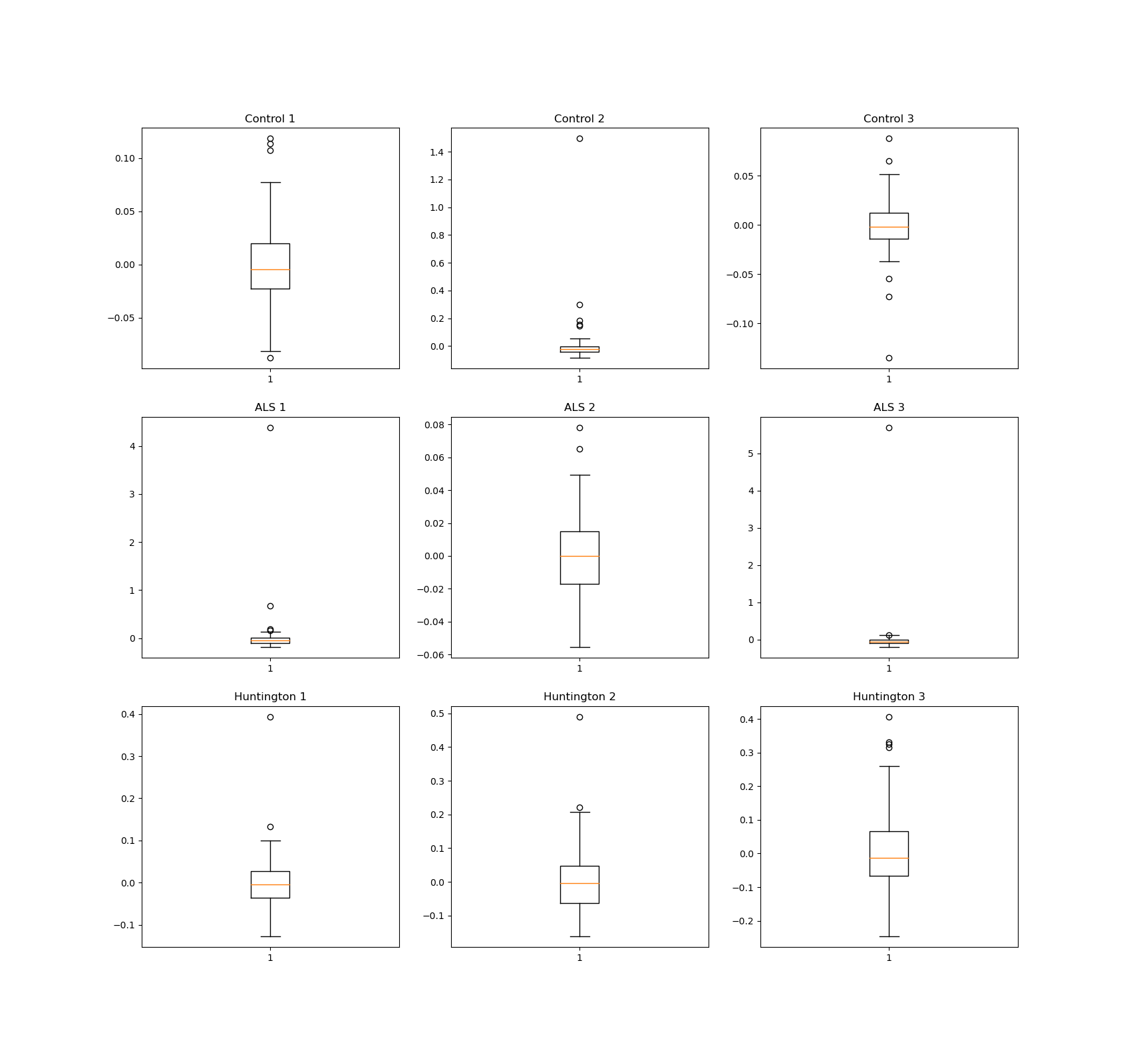}
            \caption{Box plot of \textbf{non modified} detrended stride interval time series from 9 participants in neurodegenerative disease study with 3 healthy controls, 3 ALS and 3 Huntington's}
            \label{fig:13}
        \end{figure}
        
        \begin{figure}[H]
        \centering
        \includegraphics[width=\textwidth]{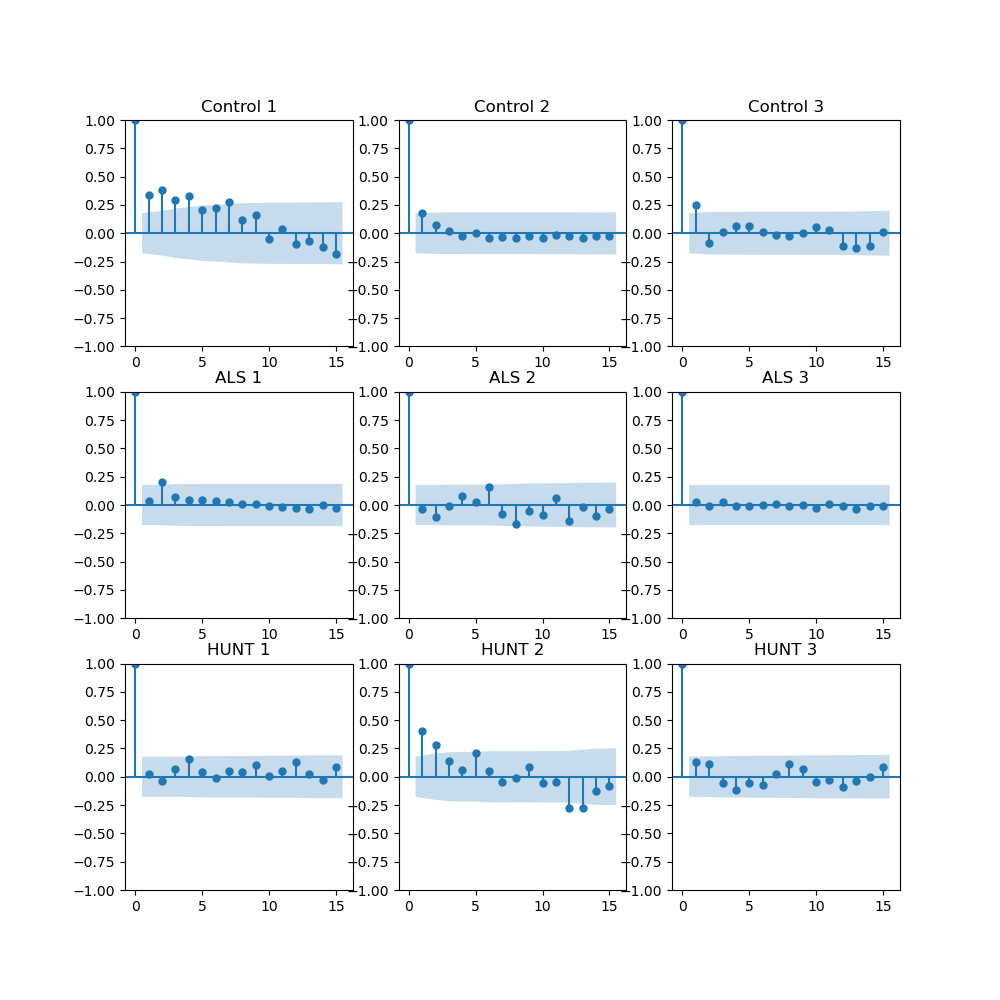}
            \caption{ACF of \textbf{non modified} detrended stride interval time series from 9 participants in neurodegenerative disease study with 3 healthy controls, 3 ALS and 3 Huntington's}\label{fig:14}
        \end{figure}
        
        This work focuses on the best classifier for linear discriminant analysis using multitaper periodogram estimation for classic analysis and modified data, and multitaper M-periodogram for robust analysis and non-modified data. As a result, the parameters of the number of multitapers ($R$) and the number of cepstrals ($L$) need to be chosen. Since the multitaper impacts only asymptotic properties, and for comparison, it will be chosen $R=7$, the same used by \cite{krafty16}.
        
        The number of cepstrals needs to be chosen in a parsimonious way, opting for the smallest possible number that still generates the highest classification rate. Table \ref{tab:5.2} contains the classification rate for the first 20 cepstrals. The maximum classification rate is achieved when $L=9$.
        
        \begin{table}[H]
            \centering
            \caption{Number of Cepstrals by classification rate}
            \label{tab:5.2}
            \begin{tabular}{c|c} 
            \hline
            \textbf{L} & \textbf{Classification rate (\%)}  \\ 
            \hline
            3          & 69.39                                \\
            4          & 69.39                                \\
            5          & 73.47                                \\
            6          & 75.51                                \\
            7          & 81.63                                \\
            8          & 81.63                                \\
            \textcolor{red}{9} & \textcolor{red}{83.67}                                \\
            10         & 75.51                                \\
            11         & 71.43                                \\
            12         & 73.47                                \\
            13         & 73.47                                \\
            14         & 75.51                                \\
            15         & 73.47                                \\
            16         & 75.51                                \\
            17         & 75.51                                \\
            18         & 79.59                                \\
            19         & 79.59                                \\
            20         & 81.63                                \\
            \hline\hline
            \end{tabular}
        \end{table}
        
        In Figure \ref{fig:15}, the horizontal axis presents the number of cepstrals, and the vertical axis shows the classification rate presented in the Table \ref{tab:5.2}. Clearly, the classification rate grows until it reaches its maximum value when $L=9$ with 83.67\% of success.
        
        \begin{figure}[H]
        \centering
        \includegraphics[width=\textwidth]{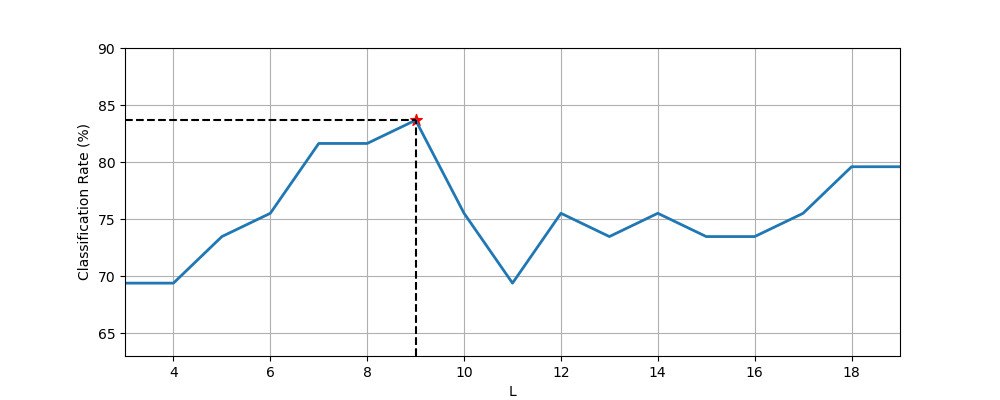}
            \caption{Number of Cepstrals by classification rate}\label{fig:15}
        \end{figure}
        
        Cepstrals are defined as the result of taking the inverse Fourier transform of the logarithm of the spectrum.   . Consequently, they carry all the information regarding not only the variability of the time series but also the variability of the spectrum, enabling a more profound and parsimonious analysis of discrimination and facilitating their use for classifying new time series.
        
        Furthermore, \cite{krafty16} states that when using multitaper periodogram estimates, most of the information is contained in the first cepstrals. According to the author, no more than $L=4$ cepstrals are necessary for this purpose. However, to achieve the classification rates shown by the author, many transformations were necessary in the original database, particularly the exclusion of 4 series from sick individuals, which contain important information for the analysis. Considering the database used in this work without excluding any sick individuals, the average classification rate across all groups is indeed equal to 69.39\%. It is worth noting that \cite{krafty16} did not report this general average rate.
        
        The classification rate can be decomposed by population based on conditional probabilities and use Bayes' formula to compute the probabilities as follows:
        
        \begin{equation}\tag{5.1}\label{eq:eq5.1}
            \rho_{ij} =  \mathbb{P}(\hat{c}_{i\ell} \in \Pi = i | \hat{c}_{j\ell} \in \Pi = j) = \frac{\mathbb{P}(\hat{c}_{j\ell} \in \Pi = i \cap \hat{c}_{j\ell} \in \Pi = j)}{\mathbb{P}(\hat{c}_{j\ell} \in \Pi = j)},
        \end{equation}
        
        \noindent where $\sum_{i=1}^{J}\rho_{ij} = 1$.
        Once in this application case $J=3$, then the $(3 \times 3)$ normalized confusion matrix can be summarized in Table \ref{tab:5.2}.
        
        \begin{table}[H]
        \centering
        \caption{Normalized Confusion Matrix for $J=3$}
        \label{tab:4}
            \begin{adjustbox}{totalheight={3.9cm}}
                \begin{tabular}{c|ccc}
                    \hline
                    \diagbox{\textbf{Predicted population}}{\textbf{Population}} & \textbf{1}        & \textbf{2}       & \textbf{3}        \\
                    \hline
                    \textbf{1}                                                & \textbf{$\rho_{11}$} & $\rho_{21}$          & $\rho_{31}$           \\
                    \textbf{2}                                                & $\rho_{12}$            & $\rho_{22}$ & $\rho_{23}$          \\
                    \textbf{3}                                                & $\rho_{31}$            & $\rho_{32}$          & $\rho_{33}$  \\
                    \hline
                    total & 1 & 1 & 1 \\
                \hline\hline
                \end{tabular}
            \end{adjustbox}
        \end{table}
        
        Note that in Table \ref{tab:4}, the main diagonal represents the correct classification rates per population, while off-diagonal entries represent the misclassification rates. The sum of each column is equal to 1, and each entry can be interpreted as the probability of being classified into each of the populations (rows) given that they belong to the known population (columns).
        
        \begin{table}[H]
            \centering
            \caption{Confusion Matrix for 49 \textbf{changed} time series of neurodegenerative disease in \%}
            \label{tab:5}
            \begin{adjustbox}{totalheight={2.6cm}}
                \begin{tabular}{c|ccc}
                \hline
                \diagbox{\textbf{Predicted population}}{\textbf{Population}} & \textbf{1}        & \textbf{2}       & \textbf{3}        \\
                \hline
                \textbf{1}                                                & \textbf{100.00\%} & 23.08\%          & 30.00\%           \\
                \textbf{2}                                                & 0.00\%            & \textbf{38.46\%} & 5.00\%            \\
                \textbf{3}                                                & 0.00\%            & 38.46\%          & \textbf{65.00\%}  \\
                \hline\hline
                \end{tabular}
            \end{adjustbox}
            \begin{flushleft}
            \small{${*}$ 1: Healthy control; 2: ALS; 3: Huntington's}
            \end{flushleft}
        \end{table}
        
        It is evident that the accuracy rate of healthy individuals reaches 100\% accuracy. However, among individuals with Sclerosis, only 30.77\% were correctly classified, with the majority being confused with individuals with Huntington's disease, resulting in a classification error rate of 38.46\%. Regarding individuals with Huntington's disease, 65\% of the 20 individuals were correctly classified, with 30\% incorrectly classified as healthy and 5\% wrongly classified as having Sclerosis.
        
        Consider now the data non-modified in the sense that the extreme values are not removed. The Assumption of normality serves as a foundation in the context of linear discriminant analysis. Cepstrals are, by definition, the Fourier transform of the logarithm of the spectrum. Additionally, \cite{priestley81} demonstrates that the periodogram follows an asymptotic chi-square distribution with two degrees of freedom. The logarithm serves as a monotone transformation that linearizes the quadratic term, and the inverse Fourier transform simply shifts the coordinates from frequency domain back to the time domain. Consequently, such variables are expected to exhibit a normal asymptotic distribution. To assess the normality of the Cepstrals for application purposes, a Shapiro test was conducted in estimated cepstrals generating a p-value of 56.18\%. This means that the normality distribution is not rejected.
        
        The confusion matrix is constructed for this data using the M-periodogram as spectral estimation and the results are presented in Table \ref{tab:6}. Exactly as the classical periodogram for modified data, the healthy control. 
        
        \begin{table}[H]
            \centering
            \caption{Confusion Matrix for 49 \textbf{non-modified} time series of neurodegenerative disease}
            \label{tab:6}
            \begin{adjustbox}{totalheight={2.7cm}}
                \begin{tabular}{c|ccc}
                \hline
                \diagbox{\textbf{Predicted population}}{\textbf{Population}} & \textbf{1}       & \textbf{2}       & \textbf{3}        \\
                \hline
                \textbf{1}                                                & \textbf{100\%} & 15.38\%          & 15.00\%               \\
                \textbf{2}                                                & 0.00\%           & \textbf{76.92\%} & 10.00\%           \\
                \textbf{3}                                                & 0.00\%           & 7.69\%           & \textbf{75.00\%}  \\
                \hline\hline
                \end{tabular}
            \end{adjustbox}
            \begin{flushleft}
            \small{${*}$ 1: control; 2: ELA; 3: Huntington's}
            \end{flushleft}
        \end{table}
        
        Figure \ref{fig:9} displays the estimated discriminant coefficients obtained through the robust proposed approach and the classical proposed by Krafty et al. (2016). 
        
        \begin{figure}[H]
        \centering
        \includegraphics[width=\textwidth]{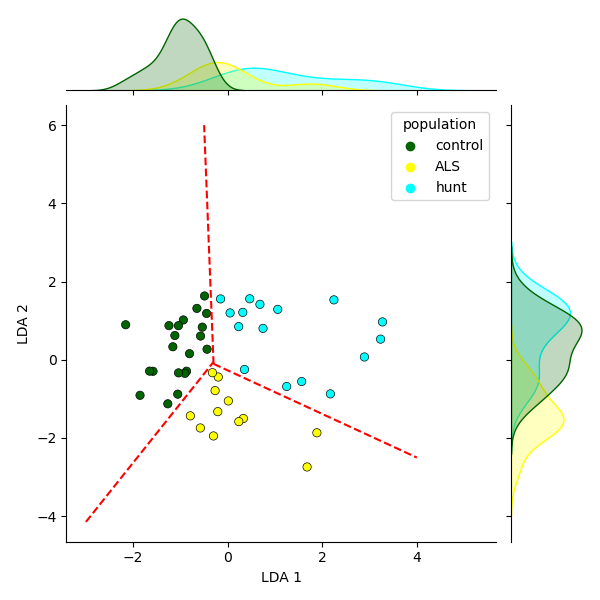}
            \caption{Linear discriminant analysis using two discriminants for \textbf{non-modified} data.}\label{fig:16}
        \end{figure}
        
        \begin{figure}[H]
        \centering
        \includegraphics[width=\textwidth]{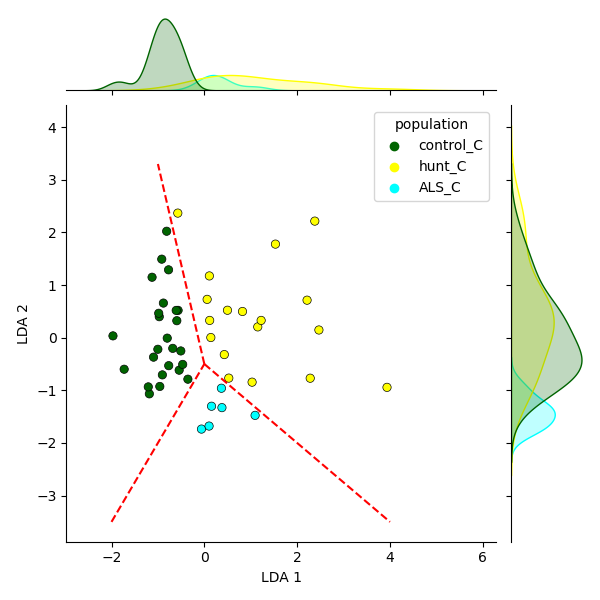}
            \caption{Linear discriminant analysis using two discriminants for \textbf{modified} data.}\label{fig:17}
        \end{figure}
        
        Specifically, coefficients estimated from the modified dataset are represented in black, while those derived from the original dataset are shown in red. Furthermore, the coefficients derived from our proposed robust method using the original dataset are presented in blue. Notably, there are distinct differences between the black and red coefficients.
        In the scenario of the adjusted dataset, the black discriminant primarily distinguishes individuals with Huntington's disease from the other two groups. The second discriminant primarily separates healthy controls from ALS patients. Conversely, in the original dataset scenario, the first discriminant predominantly separates control participants from the other two groups. In contrast, the second discriminant primarily distinguishes individuals with Huntington's disease from those with ALS.
        However, under the presence of atypical observations in the dataset, the discriminants represented in blue demonstrate better separability performance than the red coefficients. Our proposed robust method's empirical mean classification rate for the original data is approximately $75\%$, which closely aligns with the classical approach calculated for the modified dataset ($76\%$).

    \newpage
    
    \section{Conclusion}

        Time series in the time domain have proven to be ineffective for classifying data when these processes have similar systems or very close variability. Given this context, the frequency domain was chosen as the source of information to separate the series instead of the time domain. In the frequency domain, it is possible to discriminate and separate the series by frequency rather than the origin of the entire series. Additionally, noise information and the process autocorrelation structure are mixed in the spectrum due to the convolution used to transform the data from the time domain to the frequency domain. Considering the constant contribution of white noise to the total variability of the process, the most important factor for distinguishing series among populations is the autocorrelation structure. Furthermore, the spectrum of a time series always has an internal product characteristic between the autocorrelations and the noise caused by the transformation. Consequently, they can be separated using a logarithmic function, preserving all information as it is a monotonic transformation of the spectrum. Through the inverse Fourier transform, or deconvolution, this information can return to the time domain with the noise and autocorrelation information separated and summarized, referred to as cepstral coefficients. These coefficients can then be used to distinguish between processes from different populations.
    
        Additionally, this study aims to implement a model considering the spectral variability of each replica (time series) within each population. Since the processes of each population can be very similar, time series from different populations can mix, making discrimination and classification even more difficult. If the series are contaminated by an extreme value that behaves as an additive outlier, the entire correlation structure of the process is lost, along with its predictive capacity for discrimination and classification among populations.
        
        To address this challenge, this study proposed a new model that considers the spectral variability of the replicas within each population and obtains robust cepstral estimates using the M-periodogram estimator. This approach allows the model to statistically accommodate contamination by abrupt values. In this study, the robust cepstrum is referred to as M-cepstrum.
        
        In the context of asymptotic analysis, the asymptotic results of the M-periodogram can be used to show that the cepstral coefficients also tend to converge as the number of replicas ($n_j$), the length of the series ($N$), and the number of cepstral coefficients ($L$) increase. Additionally, it was shown that the cepstral coefficients can be truncated to a fixed value for the M-cepstrum, similar to what was shown by \cite{krafty16} for the cepstrum based on the classical periodogram. Empirically, the mean squared errors of both estimators tend to converge to the true value of the cepstrum as the sample size increases $N$, confirming the results demonstrated conduct the same analysis increasing $R$, $n_j$, and $L$.
        
        Once the M-periodogram estimates are obtained, it was demonstrated that M-cepstra can be used for discrimination and classification using linear discriminant analysis based on \cite{fisher36}'s method. It was observed that, through this method, a significantly smaller dimension can be used for this purpose, as the echoes containing most of the time series information are the closest in terms of autocorrelation. This information is contained in the first M-cepstral coefficients.
        
        To validate the effectiveness of the model, various types of simulations were conducted, with a focus on Monte Carlo simulations to analyze the asymptotic properties of classification rates. These simulations demonstrated the effect of extreme information loss in the classifier using the classical periodogram when the series is contaminated by an abrupt value. The M-cepstrum proved to be more effective in discriminating and classifying the time series in this context, especially for larger series sizes $N$, more replicas $n_j$, and a wide range of white noise variance possibilities. Although the robust model proposed in this study has lower classification rates in the absence of extreme values, it performs better with the presence of such values, which behave as additive outliers, providing important information for time series discrimination and classification.
        
        A practical application was conducted to evaluate the proposed M-cepstrum using gait cycle data from healthy individuals, those with amyotrophic lateral sclerosis, and those with Huntington's disease. It was shown that, in the same context used by \cite{krafty16}, fewer transformations were necessary, resulting in more satisfactory outcomes for determining whether an individual has a disease and, if so, identifying the type of disease.
        
        For future work, there are several perspectives for theoretical, applied, and computational advancements. In the theoretical and applied contexts, as proposed by \cite{krafty16}, data from the right foot sensor of individuals are used for analysis. Additionally, information from heart rate and electroencephalnrams could also be used for the same time units as a source of information. However, the database contains several other sources of information, such as data from the left foot. In this context, more than one piece of information can be used for each time unit, resulting in a multivariate time series in the context of unit cuts. In other words, more than one measure per unit of time. Instead of using only one piece of information per individual at each time, multiple pieces of information can be used in the analysis. The model should include not only the spectra of the left and right foot series but also the cross-spectrum of these two series. Consequently, a multivariate model would emerge for the cepstra, with information not only from the cepstral coefficients but also a matrix of coefficients containing both these coefficients and a cross-cepstrum coefficient. In this context, the M-cepstrum can be developed into a multivariate vector. This could result in higher classification rates for the identification and monitoring of neurodegenerative diseases.
        
        In terms of computational advancements, it is noteworthy that the M-periodogram is currently calculated using a robust regression method. In this case, a matrix of the type $X'X$ must be calculated with an internal product of order $N^2$. However, it is already known that, in pairs, the vectors of the $X$ matrix are orthogonal due to the orthogonal complement characteristics of Fourier series. Consequently, several of these products are being calculated unnecessarily, making the estimations and simulations extremely slow. This issue, besides hindering Monte Carlo simulations, limited the progress of the Shiny app developed in this study and, consequently, the practical use of the technique suggested. In the classical model estimation process, this problem is mitigated by using the \textit{Fast Fourier Transform} (FFT) algorithm, which, under certain conditions regarding the series length, yields the same result. In this study, the FFT can be used if the regression residual does not exceed the Huber constant in absolute value. This change could make the model significantly faster and more efficient, and consequently more useful for the practical use of time series.

\bibliographystyle{apalike}

\section{Appendix: PROOFS}


    \subsection{Proof of Proposition 1:}

        \begin{proof}
            
            \noindent Let $\raisebox{.5pt}{\textcircled{\raisebox{-.9pt} {A}}} = ln\left\{1+\eta_{i,jk}^2+2\eta_{i,jk}cos(\lambda)\right\}$ and $\raisebox{.5pt}{\textcircled{\raisebox{-.9pt} {B}}} = ln\left\{1+\zeta_{r,jk}^2+2\zeta_{r,jk}cos(\lambda)\right\}$ for all $i=1,...,q$ and $r=1,...,p$. Then:
                
            \vspace{.5cm}
            \noindent Applying Euler's Formula to cosine term in $\raisebox{.5pt}{\textcircled{\raisebox{-.9pt} {A}}}$:
            \vspace{.5cm}
            
            \noindent $\raisebox{.5pt}{\textcircled{\raisebox{-.9pt} {A}}} = ln\left\{ 1+\eta_{i,jk}^2+\Cancel[red]{2}\eta_{i,jk}\left[\frac{e^{i\lambda}+e^{-i\lambda}}{\Cancel[red]{2}}\right] \right\} = ln\left\{ 1+\eta_{jk}^2+\eta_{i,jk}e^{i\lambda}+\eta_{i,jk}e^{-i\lambda} \right\}$
    
            \vspace{.3cm}
            \noindent $\raisebox{.5pt}{\textcircled{\raisebox{-.9pt} {A}}}=ln\left\{\left( 1+\eta_{i,jk}e^{i\lambda} \right) \left( 1+\eta_{i,jk}e^{-i\lambda} \right)  \right\} = \left\{ln\left( 1+\eta_{i,jk}e^{i\lambda} \right) + ln\left( 1+\eta_{i,jk}e^{-i\lambda} \right)\right\}$
    
            \vspace{.3cm}
            \noindent $\raisebox{.5pt}{\textcircled{\raisebox{-.9pt} {A}}}=\raisebox{.5pt}{\textcircled{\raisebox{-.9pt} {C}}}+\raisebox{.5pt}{\textcircled{\raisebox{-.9pt} {D}}}$
        
            \vspace{.5cm}
            \noindent Using Taylor expansion for $ln(1+x) = x-\frac{x^2}{2}+\frac{x^3}{3}-\frac{x^4}{4}+ \cdots$
            \vspace{.5cm}
        
            $\raisebox{.5pt}{\textcircled{\raisebox{-.9pt} {C}}}= ln\left( 1+\eta_{i,jk}e^{i\lambda} \right) = (\eta_{i,jk}e^{i\lambda}) - \frac{(\eta_{i,jk}e^{i\lambda})^2}{2} + \frac{(\eta_{i,jk}e^{i\lambda})^3}{3} - \frac{(\eta_{i,jk}e^{i\lambda})^4}{4}+ \cdots$

            \vspace{.5cm}
            
            $=\sum_{\ell=1}^{\infty} \frac{(-1)^{\ell+1}\eta_{i,jk}^{\ell}e^{i\lambda \ell}}{\ell}$
    
            \vspace{.3cm}
            $\raisebox{.5pt}{\textcircled{\raisebox{-.9pt} {D}}}= ln\left( 1+\eta_{i,jk}e^{-i\lambda} \right) =(\eta_{i,jk}e^{-i\lambda}) - \frac{(\eta_{i,jk}e^{-i\lambda})^2}{2} + \frac{(\eta_{i,jk}e^{-i\lambda})^3}{3} - \frac{(\eta_{i,jk}e^{-i\lambda})^4}{4}+ \cdots$ 

            \vspace{.5cm}
            
            $=\sum_{\ell=1}^{\infty} \frac{(-1)^{\ell+1}\eta_{i,jk}^{\ell}e^{-i\lambda \ell}}{\ell}$
        
            \vspace{.5cm}
            \noindent Therefore, $\raisebox{.5pt}{\textcircled{\raisebox{-.9pt} {A}}}$ results in:
            \vspace{.5cm}
        
            \noindent $\raisebox{.5pt}{\textcircled{\raisebox{-.9pt} {A}}}=\sum_{\ell=1}^{\infty} \frac{(-1)^{\ell+1}\eta_{i,jk}^{\ell}e^{i\lambda \ell}}{\ell}+\sum_{\ell=1}^{\infty} \frac{(-1)^{\ell+1}\eta_{i,jk}^{\ell}e^{-i\lambda \ell}}{\ell}$
    
            \vspace{.3cm}
            \noindent $\raisebox{.5pt}{\textcircled{\raisebox{-.9pt} {A}}}=\sum_{\ell=1}^{\infty} \frac{(-1)^{\ell+1}\eta_{i,jk}^{\ell}\left[ e^{i\lambda\ell}+e^{-i\lambda\ell} \right]}{\ell}$
    
            \vspace{.3cm}
            \noindent $\raisebox{.5pt}{\textcircled{\raisebox{-.9pt} {A}}}=2\sum_{\ell=1}^{\infty} \frac{(-1)^{\ell+1}\eta_{i,jk}^{\ell}}{\ell}cos(\lambda \ell)$
    
            \vspace{.5cm}

            Following the same logic, it is possible to show that 

            \vspace{.5cm}
            \noindent $\raisebox{.5pt}{\textcircled{\raisebox{-.9pt} {B}}}=2\sum_{\ell=1}^{\infty} \frac{(-1)^{\ell+1}\zeta_{r,jk}^{\ell}}{\ell}cos(\lambda \ell)$

            \vspace{.5cm}

            Summing $\raisebox{.5pt}{\textcircled{\raisebox{-.9pt} {A}}}$ and $\raisebox{.5pt}{\textcircled{\raisebox{-.9pt} {B}}}$ we have

            \vspace{.5cm}

            $\raisebox{.5pt}{\textcircled{\raisebox{-.9pt} {A}}}+ \raisebox{.5pt}{\textcircled{\raisebox{-.9pt} {B}}} = 2\sum_{\ell=1}^{\infty} \frac{(-1)^{\ell+1}\eta_{i,jk}^{\ell}}{\ell}cos(\lambda \ell)+  2\sum_{\ell=1}^{\infty} \frac{(-1)^{\ell+1}\zeta_{r,jk}^{\ell}}{\ell}cos(\lambda \ell)$

            \vspace{.5cm}

            $=2\sum_{\ell=1}^{\infty} \frac{1}{\ell}\left((-1)^{\ell+1}\eta_{i,jk}^{\ell} + (-1)^{\ell+1}\zeta_{r,jk}^{\ell}\right)cos(\lambda \ell)$

            \vspace{.5cm}

            The first part \textit{i)} is proven. And the terms of the series in the last sum are $c_{\ell}$ for each fixed $\ell$ as described in second part of the corollary \textit{ii)}.
            
        \end{proof}

    \subsection{Proof of Corollary 1:} 

        \begin{proof}
            
            \vspace{.5cm}
                
            \noindent $\raisebox{.5pt}{\textcircled{\raisebox{-.9pt} {A}}} = ln\left\{1+\theta_{jk}^2+2\theta_{jk}cos(\lambda)\right\}$
            
            \vspace{.5cm}
            \noindent Applying Euler's Formula to cosine term:
            \vspace{.5cm}
            
            \noindent $\raisebox{.5pt}{\textcircled{\raisebox{-.9pt} {A}}} = ln\left\{ 1+\theta_{jk}^2+\Cancel[red]{2}\theta_{jk}\left[\frac{e^{i\lambda}+e^{-i\lambda}}{\Cancel[red]{2}}\right] \right\} = ln\left\{ 1+\theta_{jk}^2+\theta_{jk}e^{i\lambda}+\theta_{jk}e^{-i\lambda} \right\}$
    
            \vspace{.3cm}
            \noindent $\raisebox{.5pt}{\textcircled{\raisebox{-.9pt} {A}}}=ln\left\{\left( 1+\theta_{jk}e^{i\lambda} \right) \left( 1+\theta_{jk}e^{-i\lambda} \right)  \right\} = \left\{ln\left( 1+\theta_{jk}e^{i\lambda} \right) + ln\left( 1+\theta_{jk}e^{-i\lambda} \right)\right\}$
    
            \vspace{.3cm}
            \noindent $\raisebox{.5pt}{\textcircled{\raisebox{-.9pt} {A}}}=\raisebox{.5pt}{\textcircled{\raisebox{-.9pt} {B}}}+\raisebox{.5pt}{\textcircled{\raisebox{-.9pt} {C}}}$
        
            \vspace{.5cm}
            \noindent Using Taylor expansion for $ln(1+x) = x-\frac{x^2}{2}+\frac{x^3}{3}-\frac{x^4}{4}+ \cdots$
            \vspace{.5cm}
        
            $\raisebox{.5pt}{\textcircled{\raisebox{-.9pt} {B}}}= ln\left( 1+\theta_{jk}e^{i\lambda} \right) = (\theta_{jk}e^{i\lambda}) - \frac{(\theta_{jk}e^{i\lambda})^2}{2} + \frac{(\theta_{jk}e^{i\lambda})^3}{3} - \frac{(\theta_{jk}e^{i\lambda})^4}{4}+ \cdots =\sum_{\ell=1}^{\infty} \frac{(-1)^{\ell+1}\theta_{jk}^{\ell}e^{i\lambda \ell}}{\ell}$
    
            \vspace{.3cm}
            $\raisebox{.5pt}{\textcircled{\raisebox{-.9pt} {C}}}= ln\left( 1+\theta_{jk}e^{-i\lambda} \right) =(\theta_{jk}e^{-i\lambda}) - \frac{(\theta_{jk}e^{-i\lambda})^2}{2}$
            
            $ + \frac{(\theta_{jk}e^{-i\lambda})^3}{3} - \frac{(\theta_{jk}e^{-i\lambda})^4}{4}+ \cdots =\sum_{\ell=1}^{\infty} \frac{(-1)^{\ell+1}\theta_{jk}^{\ell}e^{-i\lambda \ell}}{\ell}$
        
            \vspace{.5cm}
            \noindent Therefore, $\raisebox{.5pt}{\textcircled{\raisebox{-.9pt} {A}}}$ results in:
            \vspace{.5cm}
        
            \noindent $\raisebox{.5pt}{\textcircled{\raisebox{-.9pt} {A}}}=\sum_{\ell=1}^{\infty} \frac{(-1)^{\ell+1}\theta_{jk}^{\ell}e^{i\lambda \ell}}{\ell}+\sum_{\ell=1}^{\infty} \frac{(-1)^{\ell+1}\theta_{jk}^{\ell}e^{-i\lambda \ell}}{\ell}$
    
            \vspace{.5cm}
            \noindent $\raisebox{.5pt}{\textcircled{\raisebox{-.9pt} {A}}}=\sum_{\ell=1}^{\infty} \frac{(-1)^{\ell+1}\theta_{jk}^{\ell}\left[ e^{i\lambda\ell}+e^{-i\lambda\ell} \right]}{\ell}$
    
            \vspace{.5cm}
            \noindent $\raisebox{.5pt}{\textcircled{\raisebox{-.9pt} {A}}}=2\sum_{\ell=1}^{\infty} \frac{(-1)^{\ell+1}\theta_{jk}^{\ell}}{\ell}cos(\lambda \ell)$
    
            \vspace{.5cm}

            \noindent $\raisebox{.5pt}{\textcircled{\raisebox{-.9pt} {A}}}=2\left(\sum_{\ell=1}^{\infty} \frac{1}{\ell} (-1)^{\ell+1}\theta_{jk}^{\ell}cos(\lambda \ell)\right)d$
    
            \vspace{.5cm}

            The first part \textit{i)} is proven. And the terms of the series in the last sum are $c_{\ell}$ for each fixed $\ell$ as described in second part of the corollary \textit{ii)}.
        \end{proof}
        
        \subsection{Proof of Corollary 2:} 
            
            \begin{proof}
            \vspace{.5cm}
            \noindent $\raisebox{.5pt}{\textcircled{\raisebox{-.9pt} {A}}} = 
                ln\left\{\frac{1}{1+\phi_{jk}^2-2\phi_{jk}cos(\lambda)}\right\}$
    
            \vspace{.3cm}
            \noindent $\raisebox{.5pt}{\textcircled{\raisebox{-.9pt} {A}}}= ln\left\{1+\phi_{jk}^2-2\phi_{jk}cos(\lambda)\right\}^{-1} = -ln\left\{1+\phi_{jk}^2-2\phi_{jk}cos(\lambda)\right\}$
            
            \vspace{.5cm}
            \noindent Applying Euler's Formula to cosine term:
            \vspace{.5cm}
    
            \noindent $\raisebox{.5pt}{\textcircled{\raisebox{-.9pt} {A}}} = -ln\left\{ 1+\phi_{jk}^2-\Cancel[red]{2}\phi_{jk}\left[\frac{e^{i\lambda}+e^{-i\lambda}}{\Cancel[red]{2}}\right] \right\} = -ln\left\{ 1+\phi_{jk}^2-\phi_{jk}e^{i\lambda}+\phi_{jk}e^{-i\lambda} \right\}$
    
            \vspace{.3cm}
            \noindent $\raisebox{.5pt}{\textcircled{\raisebox{-.9pt} {A}}}=-ln\left\{\left( 1-\phi_{jk}e^{i\lambda} \right) \left( 1-\phi_{jk}e^{-i\lambda} \right)  \right\} = - \left\{ln\left( 1-\phi_{jk}e^{i\lambda} \right) + ln\left( 1-\phi_{jk}e^{-i\lambda} \right)\right\}$
    
            \vspace{.3cm}
            \noindent $\raisebox{.5pt}{\textcircled{\raisebox{-.9pt} {A}}}=-\left[\raisebox{.5pt}{\textcircled{\raisebox{-.9pt} {B}}}+\raisebox{.5pt}{\textcircled{\raisebox{-.9pt} {C}}}\right]$
        
            \vspace{.5cm}
            \noindent Using Taylor expansion for $ln(1-x) = -x-\frac{x^2}{2}-\frac{x^3}{3}-\frac{x^4}{4}- \cdots$
            \vspace{.5cm}
        
            $\raisebox{.5pt}{\textcircled{\raisebox{-.9pt} {B}}}= ln\left( 1-\phi_{jk}e^{i\lambda} \right) = -(\phi_{jk}e^{i\lambda}) - \frac{(\phi_{jk}e^{i\lambda})^2}{2} - \frac{(\phi_{jk}e^{i\lambda})^3}{3} - \frac{(\phi_{jk}e^{i\lambda})^4}{4}- \cdots = -\sum_{\ell=1}^{\infty} \frac{\phi_{jk}^{\ell}e^{i\lambda \ell}}{\ell}$
    
            \vspace{.3cm}
            $\raisebox{.5pt}{\textcircled{\raisebox{-.9pt} {C}}}= ln\left( 1-\phi_{jk}e^{-i\lambda} \right) = -(\phi_{jk}e^{-i\lambda}) - $
            
            $ \frac{(\phi_{jk}e^{-i\lambda})^2}{2} - \frac{(\phi_{jk}e^{-i\lambda})^3}{3} - \frac{(\phi_{jk}e^{-i\lambda})^4}{4} \cdots = -\sum_{\ell=1}^{\infty} \frac{\phi_{jk}^{\ell}e^{-i\lambda \ell}}{\ell}$
        
            \vspace{.5cm}
            \noindent Therefore, $\raisebox{.5pt}{\textcircled{\raisebox{-.9pt} {A}}}$ results in:
            \vspace{.5cm}
        
            \noindent $\raisebox{.5pt}{\textcircled{\raisebox{-.9pt} {A}}}=-\left[-\sum_{\ell=1}^{\infty} \frac{\phi_{jk}^{\ell}e^{i\lambda \ell}}{\ell}-\sum_{\ell=1}^{\infty} \frac{\phi_{jk}^{\ell}e^{-i\lambda \ell}}{\ell} \right]$
    
            \vspace{.3cm}
            \noindent $\raisebox{.5pt}{\textcircled{\raisebox{-.9pt} {A}}}=\sum_{\ell=1}^{\infty} \frac{\phi_{jk}^{\ell}e^{i\lambda \ell}}{\ell}+\sum_{\ell=1}^{\infty} \frac{\phi_{jk}^{\ell}e^{-i\lambda \ell}}{\ell}$
    
            \vspace{.3cm}
            \noindent $\raisebox{.5pt}{\textcircled{\raisebox{-.9pt} {A}}}=\sum_{\ell=1}^{\infty} \frac{\phi_{jk}^{\ell}\left[ e^{i\lambda\ell}+e^{-i\lambda\ell} \right]}{\ell}$
    
            \vspace{.3cm}
            \noindent $\raisebox{.5pt}{\textcircled{\raisebox{-.9pt} {A}}}=2\sum_{\ell=1}^{\infty} \frac{\phi_{jk}^{\ell}}{\ell}cos(\lambda \ell)$
            
            \vspace{.5cm}

            \noindent $\raisebox{.5pt}{\textcircled{\raisebox{-.9pt} {A}}}=2\left(\sum_{\ell=1}^{\infty} \frac{1}{\ell}\phi_{jk}^{\ell}cos(\lambda \ell)\right)$
            
            \vspace{.5cm}
    
            The first part \textit{i)} is proven. And the terms of the series in the last sum are $c_{\ell}$ for each fixed $\ell$ as described in second part of the corollary \textit{ii)}.
        \end{proof}

    \subsection{Proof of Proposition 2:}

        \begin{proof} 
            Suppose that $\pmb{y} = B^{1/2}\pmb{x}$ a changing coordinates, which implies $\pmb{x} = B^{-1/2}\pmb{y}$. Also, suppose that $D_A$ and $D_B$ are diagonal matrices with of $A$ and $B$, with its respective eigenvalues in its principal diagonal. Additionally, suppose that $P_A = P_B$, where $P_A$ and $P_B$ are the respective eigenvector matrices, associated to respective eigenvalues. Then,

            \begin{equation*}
                \pmb{x_1} = \underset{\pmb{x} \in \mathbb{R}^{\mathbb{L}}}{\operatorname{argmax}}
                \left\{
                \frac{\pmb{x}^T A \pmb{x}}{\pmb{x}^T B \pmb{x}} 
                \right\}. 
            \end{equation*}

            \begin{equation*}
                = \underset{\pmb{y} \in \mathbb{R}^{\mathbb{L}}}{\operatorname{argmax}}
                \left\{
                \frac{(B^{-1/2}\pmb{y})^T A (B^{-1/2}\pmb{y})}{(B^{-1/2}\pmb{y})^T B (B^{-1/2}\pmb{y})} 
                \right\}. 
            \end{equation*}

            \begin{equation*}
                = \underset{\pmb{y} \in \mathbb{R}^{\mathbb{L}}}{\operatorname{argmax}}
                \left\{
                \frac{\pmb{y}^TB^{-1/2} A B^{-1/2}\pmb{y}}{\pmb{y}^TB^{-1/2} B B^{-1/2}\pmb{y}} 
                \right\}. 
            \end{equation*}

            \begin{equation*}
                = \underset{\pmb{y} \in \mathbb{R}^{\mathbb{L}}}{\operatorname{argmax}}
                \left\{
                \frac{\pmb{y}^TB^{-1/2} A B^{-1/2}\pmb{y}}{\pmb{y}^TB^{-1/2} B^{1/2}B^{1/2} B^{-1/2}\pmb{y}} 
                \right\}. 
            \end{equation*}

            \begin{equation*}
                = \underset{\pmb{y} \in \mathbb{R}^{\mathbb{L}}}{\operatorname{argmax}}
                \left\{
                \frac{\pmb{y}^TB^{-1/2} A B^{-1/2}\pmb{y}}{\pmb{y}^T\pmb{y}} 
                \right\}. 
            \end{equation*}

            \begin{equation*}
                = \underset{\pmb{y} \in \mathbb{R}^{\mathbb{L}}}{\operatorname{argmax}}
                \left\{\frac{\pmb{y}^T(P_BD_B^{-1/2}P^T_B) (P_AD_AP^T_A) (P_BD_B^{-1/2}P^T_B)\pmb{y}}{\pmb{y}^T\pmb{y}} 
                \right\}. 
            \end{equation*}

            \begin{equation*}
                = \underset{\pmb{y} \in \mathbb{R}^{\mathbb{L}}}{\operatorname{argmax}}
                \left\{\frac{\pmb{y}^TP_BD_B^{-1/2}D_AD_B^{-1/2}P^T_B\pmb{y}}{\pmb{y}^T\pmb{y}} 
                \right\}. 
            \end{equation*}

            \begin{equation*}
                \equiv \underset{\pmb{y} \in \mathbb{R}^{\mathbb{L}}}{\operatorname{argmax}}
                \left\{\frac{\pmb{y}^TP_BD_B^{-1}D_AP^T_B\pmb{y}}{\pmb{y}^T\pmb{y}} 
                \right\}. 
            \end{equation*}

            \begin{equation*}
                = \underset{\pmb{y} \in \mathbb{R}^{\mathbb{L}}}{\operatorname{argmax}}
                \left\{\frac{\pmb{y}^T(P_BD_B^{-1})(D_A)P^T_B\pmb{y}}{\pmb{y}^T\pmb{y}} 
                \right\}. 
            \end{equation*}

            \begin{equation*}
                = \underset{\pmb{y} \in \mathbb{R}^{\mathbb{L}}}{\operatorname{argmax}}
                \left\{\frac{\pmb{y}^T(B^{-1}P_B)(P_A^TAP_A)P^T_B\pmb{y}}{\pmb{y}^T\pmb{y}} 
                \right\}. 
            \end{equation*}

            \begin{equation*}
                = \underset{\pmb{y} \in \mathbb{R}^{\mathbb{L}}}{\operatorname{argmax}}
                \left\{\frac{\pmb{y}^TB^{-1}A\pmb{y}}{\pmb{y}^T\pmb{y}} 
                \right\} = \pmb{y_1}.
            \end{equation*}

    \end{proof}

    \begin{lemma}
        As a direct consequence of proposition 4, $E(I^{MR}_{jk}(\lambda_m)) \approx E(I^{R}_{jk}(\lambda_m)) = S_X(\lambda_m)$
    \end{lemma}

    \subsection{Proof of Theorem 2:} 
        \begin{proof} 
            After constructing the multitaper M-periodogram through an averaging process of periodograms derived from orthogonal data sine tapers, the proof can proceed directly by applying Chebyshev inequality. So, consider that the for each $r$, the periodograms are independent due to orthogonality of the tapers and under \ref{A.1}, $ E \lvert X_{jkt} \rvert^4 < \infty \Rightarrow Var(I^{rM}_{jk}(\lambda_m)) < \infty$ when $R = R_N\rightarrow \infty$ (as defined in \ref{remark 4}. Also, from Equation \ref{eq:eq3.3},
            
            $$ I^{MR}_{jk}(\lambda_m) = \frac{1}{R} \sum_{r=1}^R I^{rM}_{jk}(\lambda_m) $$ 

            Since \( I^{rM}_{jk}(\lambda_m) \) are independent with finite variance, the variance of the mean \( I^{MR}_{jk}(\lambda_m) \) is:
            $$ \text{Var}\left(I^{MR}_{jk}(\lambda_m)\right) = \text{Var}\left(\frac{1}{R} \sum_{r=1}^R I^{rM}_{jk}(\lambda_m)\right) $$

            $$ = \frac{1}{R^2} \sum_{r=1}^R \text{Var}\left(I^{rM}_{jk}(\lambda_m)\right) $$

            For each $r$, let the finite variance be \( \text{Var}(I^{rM}_{jk}(\lambda_m)) = \sigma_r^2 \). Then, for any given $\epsilon > 0$ and
        
            $$\mathbb{P}\left( \left \lvert I^{MR}_{jk}(\lambda_m) - S_X(\lambda_m) \right \rvert > \epsilon \right) = \mathbb{P}\left(\left \lvert I^{MR}_{jk}(\lambda_m) - S_X(\lambda_m) \right \rvert^2 > \epsilon^2 \right) \leq \frac{\text{Var}\left[I^{MR}_{jk}(\lambda_m)\right]}{\epsilon^2}$$
        
            $$= \frac{\sigma_r^2}{R\epsilon^2}   \overset{\text{p}}{\longrightarrow} 0\text{, when} \quad R \rightarrow \infty $$ 
            
            Note that under Assumption\ref{A.1}, the fourth
            moment of $X_{jkt}$ is bounded, then $E\lvert X_{jkt} \rvert^4 < \infty \implies var(I^{Mr}_{jk}(\lambda_m)) < \infty$ for fixed $j,k,\lambda_m$.

        \end{proof}

    \subsection{Proof of Lemma 2:} 
        \begin{proof}
            Under the theorem \ref{theorem1} and the continuous mapping theorem, $E[I^M(\lambda_m)] \rightarrow S_X(\lambda_m) \Longrightarrow E\left\{ln[I^M(\lambda_m)]\right\} \rightarrow ln[S_X(\lambda_m)]$.
        \end{proof}
        
    \subsection{Proof of Proposition 4:} 
        \begin{proof}
        \begin{equation}
            \left \lvert \hat{\gamma}^M(\tau) - \hat{\gamma}(\tau) \right \rvert  = \left \lvert \hat{\gamma}^M(\tau)  - \hat{\gamma}(\tau) +\gamma(\tau) - \gamma(\tau)\right \rvert 
        \end{equation}

        \begin{equation}
            = \left \lvert \hat{\gamma}^M(\tau) - \gamma(\tau)+\gamma(\tau)  - \hat{\gamma}(\tau)  \right \rvert
        \end{equation}

        \begin{equation}
            \leq \left \lvert \hat{\gamma}^M(\tau) -  \gamma(\tau)\right \rvert + \left \lvert \gamma(\tau) - \hat{\gamma}(\tau) \right \rvert = 0
        \end{equation}
        
        The last term in the right term of the inequality is is zero as showed in \cite{celine22} and the first therm because 

        \begin{equation}
            \lim_{N \to \infty} \left \lvert \hat{\gamma}_N^M(\tau)-\hat{\gamma}(\tau) \right \rvert = \left \lvert \gamma_{\psi}(\tau) - \gamma(\tau) \right \rvert
        \end{equation}

        \begin{equation}
            \leq \lvert a \gamma(\tau) - \gamma(\tau) \rvert = \lvert (1-a) \gamma(\tau) \rvert   = 0 \text{, when} \quad a \rightarrow 1
        \end{equation}

        The final result for (3) shows that $\left \lvert \hat{\gamma}^M(\tau) - \hat{\gamma}(\tau) \right \rvert \overset{\text{p}}{\longrightarrow} 0$. And using proposition 6.1.3 in \cite{brockwell91}, when $\left \lvert \hat{\gamma}_N^M(\tau) \overset{\text{p}}{\longrightarrow} \gamma_{\psi}(\tau) \right \rvert$, $\left \lvert \hat{\gamma}_N(\tau) \overset{\text{p}}{\longrightarrow} \gamma(\tau) \right \rvert$ and (3), then $\hat{\gamma}^M(\tau) \overset{\text{p}} {\longrightarrow} \gamma(\tau)$
           
        \end{proof}

    \subsection{Proof of Theorem 3: } 
        \begin{proof} 
            For all $j$th population and $k$th replicates, if $\hat{c}_{jkL}$ is the cepstral classical periodogram based, \cite{krafty16} showed that $\lvert \hat{c}_{jkL} - c_{jk\ell} \rvert = \mathcal{O}_p(1)$ and it is easy to show that $\lvert \hat{c}^M_{jkL} - \hat{c}_{jkL} \rvert = \mathcal{O}_p(1)$. Then, by the triangular inequality we have:

            $$\lvert \hat{c}^M_{jkL} - c_{jk\ell} \rvert = \lvert \hat{c}^M_{jkL} - \hat{c}_{jkL} +\hat{c}_{jkL}- c_{jk\ell} \rvert \leq \lvert \hat{c}^M_{jkL} - \hat{c}_{jkL} \rvert + \lvert \hat{c}_{jkL}- c_{jk\ell} \rvert = \mathcal{O}_p(1)$$
            
        \end{proof}

    \subsection{Proof of Theorem 4: } 
        \begin{proof} 
            Let $\{\hat{c}_{jk\ell}^M\}_{\ell=1}^{\infty}$ is a sequence of estimated cepstral coefficients based on M-Periodogram and $\{\hat{p}_{s\ell}^M\}_{\ell=1}^{\infty}$ the truncated weights functions. The distance between discriminant estimator in the Equation \ref{eq:eq3.5} of the true discriminant function defined in the Equation \ref{eq:eq2.21} is defined as

            $$\left\lvert \hat{d}_{jks}^M - d_{jks}\right\rvert = \left \lvert \sum_{\ell = 1}^{L-1} \hat{p}^M_{jks\ell}\hat{c}^M_{jk\ell} - \sum_{\ell=1}^{\infty}p_{jks\ell}c_{jk\ell} \right\rvert$$

            Adding and subtracting $\sum_{\ell=1}^{L-1}\hat{p}^M_{jks\ell}c_{jk\ell}$ inside of the modulo on the right side we get

            $$= \left \lvert \sum_{\ell = 1}^{L-1} \hat{p}^M_{jks\ell}\hat{c}_{jk\ell} - \sum_{\ell=1}^{L-1}\hat{p}^M_{jks\ell}c_{jk\ell} + 
            \sum_{\ell=1}^{L-1}\hat{p}^M_{jks\ell}c_{jk\ell} - 
            \sum_{\ell=1}^{\infty}p_{jks\ell}c_{jk\ell} \right\rvert$$

            $$= \left \lvert \sum_{\ell = 1}^{L-1} \hat{p}^M_{jks\ell}(\hat{c}_{jk\ell} - c_{jk\ell}) + 
            \sum_{\ell=1}^{L-1}\hat{p}^M_{jks\ell}c_{jk\ell} - 
            \sum_{\ell=1}^{\infty}p_{jks\ell}c_{jk\ell} \right\rvert$$

            And using triangular inequality we get

            $$\leq \left \lvert \sum_{\ell = 1}^{L-1} \hat{p}^M_{jks\ell}(\hat{c}_{jk\ell} - c_{jk\ell})\right\rvert + 
            \left \lvert \sum_{\ell=1}^{L-1}\hat{p}^M_{jks\ell}c_{jk\ell} - 
            \sum_{\ell=1}^{\infty}p_{jks\ell}c_{jk\ell} \right\rvert$$

            Taking the limit of probability in both sides we get

            $$\lim_{L \to \infty}\mathbb{P} \left\lvert \hat{d}_{jks}^M - d_{jks}\right\rvert \leq \lim_{L \to \infty}\mathbb{P} \left \lvert \sum_{\ell = 1}^{L-1} \hat{p}^M_{jks\ell}(\hat{c}_{jk\ell} - c_{jk\ell})\right\rvert 
            +
            \lim_{L \to \infty}\mathbb{P}\left \lvert \sum_{\ell=1}^{L-1}\hat{p}^M_{jks\ell}c_{jk\ell} - 
            \sum_{\ell=1}^{\infty}p_{jks\ell}c_{jk\ell} \right\rvert$$

            $$ = \lim_{L \to \infty}\mathbb{P} \left \lvert \sum_{\ell = 1}^{L-1} \hat{p}^M_{jks\ell}(\hat{c}_{jk\ell} - c_{jk\ell})\right\rvert 
            +
            \lim_{L \to \infty}\mathbb{P}\left \lvert \sum_{\ell=1}^{L-1}(\hat{p}^M_{jks\ell}-p_{jks\ell})c_{jk\ell} \right\rvert = \mathcal{O}_p(1)$$

            The first term is due to the \ref{theorem2} and the second term is due $\hat{p}_{s\ell}$ is surprisingly to be bounded in probability.
    
        \end{proof}

\end{document}